\documentclass[12pt]{article}
\usepackage{amsmath,amsthm}
\usepackage{eufrak}
\usepackage{pb-diagram,lamsarrow,pb-lams}

\DeclareMathAlphabet{\bb}{U}{msb}{m}{n} \gdef\C{\bb C} \gdef\dZ{\bb
Z}   \gdef\dS{\bb S} \gdef\R{\bb R}
\gdef\K{\bb K} \gdef\BH{\bb H} \gdef\F{\bb F} \gdef\dO{\bb O}

\DeclareMathOperator{\End}{End} \DeclareMathOperator{\spin}{{\bf
Spin}} 
\DeclareMathOperator{\fD}{\mathfrak{D}}

\DeclareMathOperator{\Sym}{Sym} 
 
 \DeclareMathOperator{\Mat}{Mat}
 
 \DeclareMathOperator{\SL}{SL}
\DeclareMathOperator{\SO}{SO}\DeclareMathOperator{\SU}{SU}

\newcommand{\bcirc}{\raisebox{0.5mm}{$\scriptstyle\bigcirc$}}

\newcommand{\bi}{{\bf i}}
\newcommand{\cA}{\mathcal{A}}

\newcommand{\cP}{{\cal P}}

\newcommand{\cM}{{\cal M}}

\newcommand{\sA}{{\sf A}}
\newcommand{\sB}{{\sf B}}

\newcommand{\sX}{{\sf X}}
\newcommand{\sY}{{\sf Y}}

\newcommand{\bz}{{\bf z}}

\newcommand{\bZ}{{\bf Z}}

\newcommand{\fA}{\mathfrak{A}}
\newcommand{\fM}{\mathfrak{M}}
\newcommand{\fG}{\mathfrak{G}}
\newcommand{\fC}{\mathfrak{C}}
\newcommand{\fR}{\mathfrak{R}}
\newcommand{\fH}{\mathfrak{H}}

\newcommand{\fO}{\mathfrak{O}}

\newcommand{\fL}{\mathfrak{L}}

\newcommand{\fg}{\mathfrak{g}}

\newcommand{\balpha}{\boldsymbol{\alpha}}
\newcommand{\cl}{C\kern -0.2em \ell}

\newcommand{\p}{\prime}
\newcommand{\e}{\mbox{\bf e}}

\newcommand{\ld}{\left[}
\newcommand{\rd}{\right]}

\newtheorem{thm}{Theorem}
\newtheorem{defn}{Definition}
\begin{document}
\title{Cyclic structures of Cliffordian supergroups and particle representations of $\spin_+(1,3)$}
\author{V.~V. Varlamov\\
{\small\it Department of Mathematics, Siberian State Industrial
University,}\\
{\small\it Kirova 42, Novokuznetsk 654007, Russia}}
\date{}
\maketitle
\begin{abstract}
Supergroups are defined in the framework of $\dZ_2$-graded Clifford
algebras over the fields of real and complex numbers, respectively.
It is shown that cyclic structures of complex and real supergroups
are defined by Brauer-Wall groups related with the modulo 2 and
modulo 8 periodicities of the complex and real Clifford algebras.
Particle (fermionic and bosonic) representations of a universal
covering (spinor group $\spin_+(1,3)$) of the proper orthochronous
Lorentz group are constructed via the Clifford algebra formalism.
Complex and real supergroups are defined on the representation
system of $\spin_+(1,3)$. It is shown that a cyclic (modulo 2)
structure of the complex supergroup is equivalent to a
supersymmetric action, that is, it converts fermionic
representations into bosonic representations and vice versa. The
cyclic action of the real supergroup leads to a much more
high-graded symmetry related with the modulo 8 periodicity of the
real Clifford algebras. This symmetry acts on the system of real
representations of $\spin_+(1,3)$.
\end{abstract}

\section{Introduction}
It is well known that Clifford algebras (systems of hypercomplex
numbers) have a broad application in many areas of mathematical and
theoretical physics. From historical point of view, Clifford
algebras have essentially geometric origin \cite{3}, because they
are the synthesis of Hamilton quaternion calculus \cite{Ham} and
Grassmann {\it Ausdehnungslehre} \cite{Grass}, and by this reason
they called by Clifford as {\it geometric algebras} \cite{Cliff2}.
Further, Lipschitz \cite{Lips} showed that Clifford algebras are
related closely with the study of rotation groups of
multidimensional spaces. After fundamental works of Cartan
\cite{Car08}, Witt \cite{Wit37} and Chevalley \cite{Che54} the
Clifford algebra theory takes its modern form
\cite{Cru91,Port95,Lou97}. It is known that Clifford algebras
inherit $\dZ_2$-graded structure from the Grassmann algebras. This
fact allows one to define \emph{Cliffordian supergroups} using a
classical theory of formal Lie groups \cite{Boc46,BK70}. Along with
the $\dZ_2$-graded structure Clifford algebras possess modulo 2 and
modulo 8 periodicities over the fields of complex and real numbers,
respectively (see \cite{AtBSh,BTr87,BT88}). These periodic relations
are described by Brauer-Wall groups \cite{Wal64,Lou81}. The cyclic
(modulo 2 and 8) structure is the most essential property of the
Cliffordian supergroups.

In the present work we consider the field
$\boldsymbol{\psi}(\balpha)=\langle
x,\fg\,|\boldsymbol{\psi}\rangle$ on the representation spaces of a
spinor group $\spin_+(1,3)$, where $x\in T_4$ and
$\fg\in\spin_+(1,3)$ ($T_4$ is a translation subgroup of the
Poincar\'{e} group $\cP$). At this point, four parameters $x^\mu$
correspond to position of the point-like object, whereas remaining
six parameters $\bz\in\spin_+(1,3)$ define orientation in quantum
description of orientable (extended) object \cite{GS09,GS10}. In
general, the field $\boldsymbol{\psi}(\balpha)$ is defined on the
homogeneous space $\cM_{10}=\R^{1,3}\times\fL_6$ (a group manifold
of $\cP$), where $\R^{1,3}$ is the Minkowski spacetime and $\fL_6$
is a group manifold of $\SO_0(1,3)$. On the other hand, the space
$\cM_{10}=\R^{1,3}\times\fL_6$ can be understood as a fiber bundle,
where a bundle base is $\R^{1,3}$. All intrinsic properties of the
field $\boldsymbol{\psi}(\balpha)$ are described within the group
$\spin_+(1,3)$ and its representations.

The paper is organized as follows. A short introduction to the
Clifford algebra theory is given in section 2. $\dZ_2$-graded
Clifford algebras, Brauer-Wall supergroups and Trautman diagrams
(spinorial clocks) are considered in section 3 over the fields of
real and complex numbers. Complex and real representations of the
field $\boldsymbol{\psi}(\balpha)=\langle
x,\fg\,|\boldsymbol{\psi}\rangle$ of different types are constructed
in section 4 within representations of $\spin_+(1,3)$. A
relationship between tensor products of the Clifford algebras and a
Gel'fand-Naimark representation basis of the proper orthochronous
Lorentz group is established. Complex and real supergroups are
defined on the representation systems of $\spin_+(1,3)$ in section
5.

\section{Algebraic preliminaries}
In this section we will consider some basic facts concerning
Clifford algebras. Let $\F$ be a field of characteristic 0
$(\F=\R,\,\F=\C)$, where $\R$ and $\C$ are the fields of real and
complex numbers, respectively. A Clifford algebra $\cl$ over a field
$\F$ is an algebra with $2^n$ basis elements: $\e_0$ (unit of the
algebra) $\e_1,\e_2,\ldots,\e_n$ and products of the one-index
elements $\e_{i_1i_2\ldots i_k}=\e_{i_1}\e_{i_2}\ldots\e_{i_k}$.
Over the field $\F=\R$ the Clifford algebra is denoted as
$\cl_{p,q}$, where the indices $p$ and $q$ correspond to the indices
of the quadratic form
\[
Q=x^2_1+\ldots+x^2_p-\ldots-x^2_{p+q}
\]
of a vector space $V$ associated with $\cl_{p,q}$. The
multiplication law of $\cl_{p,q}$ is defined by a following rule:
\begin{equation}\label{e1}
\e^2_i=\sigma(p-i)\e_0,\quad\e_i\e_j=-\e_j\e_i,
\end{equation}
where
\begin{equation}\label{e2}
\sigma(n)=\left\{\begin{array}{rl}
-1 & \mbox{if $n\leq 0$},\\
+1 & \mbox{if $n>0$}.
\end{array}\right.
\end{equation}
The square of a volume element\index{element!volume}
$\omega=\e_{12\ldots n}$ ($n=p+q$) plays an important role in the
theory of Clifford algebras,
\begin{equation}\label{e3}
\omega^2=\left\{\begin{array}{rl}
-1 & \mbox{if $p-q\equiv 2,3,6,7\pmod{8}$},\\
+1 & \mbox{if $p-q\equiv 0,1,4,5\pmod{8}$}.
\end{array}\right.
\end{equation}
A center\index{center} $\bZ_{p,q}$ of the algebra $\cl_{p,q}$
consists of the unit $\e_0$ and the volume element $\omega$. The
element $\omega=\e_{12\ldots n}$ belongs to a center when $n$ is
odd. Indeed,
\begin{eqnarray}
\e_{12\ldots n}\e_i&=&(-1)^{n-i}\sigma(q-i)\e_{12\ldots i-1
i+1\ldots n},
\nonumber\\
\e_i\e_{12\ldots n}&=&(-1)^{i-1}\sigma(q-i)\e_{12\ldots i-1
i+1\ldots n}, \nonumber
\end{eqnarray}
therefore, $\omega\in\bZ_{p,q}$ if and only if $n-i\equiv
i-1\pmod{2}$, that is, $n$ is odd. Further, using (\ref{e3}) we
obtain
\begin{equation}\label{e4}
\bZ_{p,q}=\left\{\begin{array}{rl}
\phantom{1,}1 & \mbox{if $p-q\equiv 0,2,4,6\pmod{8}$},\\
1,\omega & \mbox{if $p-q\equiv 1,3,5,7\pmod{8}$}.
\end{array}\right.
\end{equation}

An arbitrary element $\cA$ of the algebra $\cl_{p,q}$ is represented
by a following formal polynomial
\begin{gather}
\cA=a^0\e_0+\sum^n_{i=1}a^i\e_i+\sum^n_{i=1}\sum^n_{j=1}a^{ij}\e_{ij}+
\ldots+\sum^n_{i_1=1}\cdots\sum^n_{i_k=1}a^{i_1\ldots
i_k}\e_{i_1\ldots i_k}+
\nonumber\\
+\ldots+a^{12\ldots n}\e_{12\ldots n}=\sum^n_{k=0}a^{i_1i_2\ldots
i_k} \e_{i_1i_2\ldots i_k}.\nonumber
\end{gather}
In Clifford algebra $\cl$ there exist four fundamental automorphisms.\\[0.2cm]
1) {\bf Identity}: An automorphism $\cA\rightarrow\cA$ and
$\e_{i}\rightarrow\e_{i}$.\\
This automorphism is an identical automorphism of the algebra $\cl$.
$\cA$ is an arbitrary element of $\cl$.\\[0.2cm]
2) {\bf Involution}: An automorphism $\cA\rightarrow\cA^\star$ and
$\e_{i}\rightarrow-\e_{i}$.\\
In more details, for an arbitrary element $\cA\in\cl$ there exists a
decomposition $ \cA=\cA^{\p}+\cA^{\p\p}, $ where $\cA^{\p}$ is an
element consisting of homogeneous odd elements, and $\cA^{\p\p}$ is
an element consisting of homogeneous even elements, respectively.
Then the automorphism $\cA\rightarrow\cA^{\star}$ is such that the
element $\cA^{\p\p}$ is not changed, and the element $\cA^{\p}$
changes sign: $ \cA^{\star}=-\cA^{\p}+\cA^{\p\p}. $ If $\cA$ is a
homogeneous element, then
\begin{equation}\label{auto16}
\cA^{\star}=(-1)^{k}\cA,
\end{equation}
where $k$ is a degree of the element. It is easy to see that the
automorphism $\cA\rightarrow\cA^{\star}$ may be expressed via the
volume element $\omega=\e_{12\ldots p+q}$:
\begin{equation}\label{auto17}
\cA^{\star}=\omega\cA\omega^{-1},
\end{equation}
where $\omega^{-1}=(-1)^{\frac{(p+q)(p+q-1)}{2}}\omega$. When $k$ is
odd, the basis elements $\e_{i_{1}i_{2}\ldots i_{k}}$ the sign
changes, and when $k$ is even, the sign
is not changed.\\[0.2cm]
3) {\bf Reversion}: An antiautomorphism
$\cA\rightarrow\widetilde{\cA}$ and
$\e_i\rightarrow\e_i$.\\
The antiautomorphism $\cA\rightarrow\widetilde{\cA}$ is a reversion
of the element $\cA$, that is the substitution of each basis element
$\e_{i_{1}i_{2}\ldots i_{k}}\in\cA$ by the element
$\e_{i_{k}i_{k-1}\ldots i_{1}}$:
\[
\e_{i_{k}i_{k-1}\ldots i_{1}}=(-1)^{\frac{k(k-1)}{2}}
\e_{i_{1}i_{2}\ldots i_{k}}.
\]
Therefore, for any $\cA\in\cl_{p,q}$ we have
\begin{equation}\label{auto19}
\widetilde{\cA}=(-1)^{\frac{k(k-1)}{2}}\cA.
\end{equation}
4) {\bf Conjugation}: An antiautomorphism
$\cA\rightarrow\widetilde{\cA^\star}$
and $\e_i\rightarrow-\e_i$.\\
This antiautomorphism is a composition of the antiautomorphism
$\cA\rightarrow\widetilde{\cA}$ with the automorphism
$\cA\rightarrow\cA^{\star}$. In the case of a homogeneous element
from the formulae (\ref{auto16}) and (\ref{auto19}), it follows
\begin{equation}\label{20}
\widetilde{\cA^{\star}}=(-1)^{\frac{k(k+1)}{2}}\cA.
\end{equation}

The all Clifford algebras $\cl_{p,q}$ over the field $\F=\R$ are
divided
into eight different types with a following division ring structure.\\[0.3cm]
{\bf I}. Central simple algebras.
\begin{description}
\item[1)] Two types $p-q\equiv 0,2\pmod{8}$ with a division ring
$\K\simeq\R$.
\item[2)] Two types $p-q\equiv 3,7\pmod{8}$ with a division ring
$\K\simeq\C$.
\item[3)] Two types $p-q\equiv 4,6\pmod{8}$ with a division ring
$\K\simeq\BH$.
\end{description}
{\bf II}. Semi-simple algebras.
\begin{description}
\item[4)] The type $p-q\equiv 1\pmod{8}$ with a double division ring
$\K\simeq\R\oplus\R$.
\item[5)] The type $p-q\equiv 5\pmod{8}$ with a double quaternionic
division ring $\K\simeq\BH\oplus\BH$.
\end{description}
The Table 1 (Budinich-Trautman Periodic Table \cite{BT88})
explicitly shows a distribution of the real Clifford algebras in
dependence on the division ring structure\index{structure!division
ring}, here ${}^2\R(n)=\R(n)\oplus\R(n)$ and
${}^2\BH(n)=\BH(n)\oplus\BH(n)$.
\begin{figure}\label{Periodic}
\begin{center}
{\renewcommand{\arraystretch}{1.2}
\begin{tabular}{c|ccccccccl}
p  &  0 & 1 & 2 & 3 & 4 & 5 & 6 & 7 & \ldots\\ \hline
q &      &   &   &   &   &   &   &   &\\
0 &
$\R$&${}^2\R$&$\R(2)$&$\C(2)$&$\BH(2)$&${}^2\BH(2)$&$\BH(4)$&$\C(8)$&
$\ldots$\\
1&$\C$&$\R(2)$&${}^2\R(2)$&$\R(4)$&$\C(4)$&$\BH(4)$&${}^2\BH(4)$&$\BH(8)$&
$\ldots$\\
2&$\BH$&$\C(2)$&$\R(4)$&${}^2\R(4)$&$\R(8)$&$\C(8)$&$\BH(8)$&${}^2\BH(8)$&
$\ldots$\\
3&${}^2\BH$&$\BH(2)$&$\C(4)$&$\R(8)$&${}^2\R(8)$&$\R(16)$&$\C(16)$&$\BH(16)$&
$\ldots$\\
4&$\BH(2)$&${}^2\BH(2)$&$\BH(4)$&$\C(8)$&$\R(16)$&${}^2\R(16)$&$\R(32)$&
$\C(32)$&$\ldots$\\
5&$\C(4)$&$\BH(4)$&${}^2\BH(4)$&$\BH(8)$&$\C(16)$&$\R(32)$&${}^2\R(32)$&
$\R(64)$&$\ldots$\\
6&$\R(8)$&$\C(8)$&$\BH(8)$&${}^2\BH(8)$&$\BH(16)$&$\C(32)$&$\R(64)$&
${}^2\R(64)$&$\ldots$\\
7&${}^2\R(8)$&$\R(16)$&$\C(16)$&$\BH(16)$&${}^2\BH(16)$&$\BH(32)$&$\C(64)$&
$\R(128)$&$\ldots$\\
$\vdots$&$\vdots$&$\vdots$&$\vdots$&$\vdots$&$\vdots$&$\vdots$&$\vdots$&
$\vdots$
\end{tabular}}
\end{center}
\hspace{0.3cm}
\begin{center}{\small \textbf{Tab.\,1:} Distribution of the real Clifford algebras.}
\end{center}
\end{figure}

Over the field $\F=\C$ there is an isomorphism
$\C_n\simeq\Mat_{2^{n/2}}(\C)$ and there are two different types of
complex Clifford algebras $\C_n$: $n\equiv 0\pmod{2}$ and $n\equiv
1\pmod{2}$.

When $\cl_{p,q}$ is simple, then the map
\begin{equation}\label{Simple}
\cl_{p,q}\overset{\gamma}{\longrightarrow}\End_{\K}(\dS),\quad
u\longrightarrow\gamma(u),\quad \gamma(u)\psi=u\psi
\end{equation}\begin{sloppypar}\noindent
gives an irreducible and faithful representation of $\cl_{p,q}$ in
the spinspace $\dS_{2^m}(\K)\simeq I_{p,q}=\cl_{p,q}f$, where
$\psi\in\dS_{2^m}$, $m=\frac{p+q}{2}$, $I_{p,q}$ is a minimal left
ideal of $\cl_{p,q}$.\end{sloppypar}

On the other hand, when $\cl_{p,q}$ is semi-simple, then the map
\begin{equation}\label{Semi-Simple}
\cl_{p,q}\overset{\gamma}{\longrightarrow}\End_{\K\oplus\hat{\K}}
(\dS\oplus\hat{\dS}),\quad u\longrightarrow\gamma(u),\quad
\gamma(u)\psi=u\psi
\end{equation}
gives a faithful but reducible representation of $\cl_{p,q}$ in the
double spinspace\index{spinspace!double} $\dS\oplus\hat{\dS}$, where
$\hat{\dS}=\{\hat{\psi}|\psi\in\dS\}$. In this case, the ideal
$\dS\oplus\hat{\dS}$ possesses a right $\K\oplus\hat{\K}$-linear
structure, $\hat{\K}=\{\hat{\lambda}|\lambda\in\K\}$, and
$\K\oplus\hat{\K}$ is isomorphic to the double division ring
$\R\oplus\R$ when $p-q\equiv 1\pmod{8}$ or to $\BH\oplus\BH$ when
$p-q\equiv 5\pmod{8}$. The map $\gamma$ in (\ref{Simple}) and
(\ref{Semi-Simple}) defines the so called {\it left-regular} spinor
representation\index{representation!spinor} of $\cl(Q)$ in $\dS$ and
$\dS\oplus\hat{\dS}$, respectively. Furthermore, $\gamma$ is {\it
faithful} which means that $\gamma$ is an algebra monomorphism. In
(\ref{Simple}), $\gamma$ is {\it irreducible} which means that $\dS$
possesses no proper (that is, $\neq 0,\,\dS$) invariant
subspaces\index{subspace!invariant} under the left action of
$\gamma(u)$, $u\in\cl_{p,q}$. Representation $\gamma$ in
(\ref{Semi-Simple}) is therefore {\it reducible} since
$\{(\psi,0)|\psi\in\dS\}$ and
$\{(0,\hat{\psi})|\hat{\psi}\in\hat{\dS}\}$ are two proper subspaces
of $\dS\oplus\hat{\dS}$ invariant under $\gamma(u)$ (see
\cite{Lou97,Cru91,Port95}).

\section{$\dZ_2$-graded Clifford algebras and Brauer-Wall groups}
\label{Sec:1.5} The algebra $\cl$ is naturally $\dZ_2$-graded. Let
$\cl^+$ (correspondingly $\cl^-$) be a set consisting of all even
(correspondingly odd) elements of the algebra $\cl$. The set $\cl^+$
is a subalgebra of $\cl$. It is obvious that $\cl=\cl^+\oplus\cl^-$,
and also $\cl^+\cl^+ \subset\cl^+,\,\cl^+\cl^-\subset\cl^-,\,
\cl^-\cl^+\subset\cl^-,\,\cl^-\cl^-\subset \cl^+$. A degree $\deg a$
of the even (correspondingly odd) element $a\in\cl$ is equal to 0
(correspondingly 1). Let $\mathfrak{A}$ and $\mathfrak{B}$ be the
two associative $\dZ_2$-graded algebras over the field $\F$; then a
multiplication of homogeneous elements\index{element!homogeneous}
$\mathfrak{a}^\prime\in\mathfrak{A}$ and
$\mathfrak{b}\in\mathfrak{B}$ in a graded tensor
product\index{product!graded tensor}
$\mathfrak{A}\hat{\otimes}\mathfrak{B}$ is defined as follows:
$(\mathfrak{a}\otimes \mathfrak{b})(\mathfrak{a}^\prime \otimes
\mathfrak{b}^\prime)=(-1)^{\deg\mathfrak{b}\deg\mathfrak{a}^\prime}
\mathfrak{a}\mathfrak{a}^\prime\otimes\mathfrak{b}\mathfrak{b}^\prime$.
\begin{thm}[{\rm Chevalley \cite{Che55}}]
Let $V$ and $V^\prime$ are vector spaces over the field $\F$ and let
$Q$ and $Q^\prime$ are quadratic forms\index{form!quadratic} for $V$
and $V^\prime$. Then a Clifford algebra $\cl(V\oplus
V^\prime,Q\oplus Q^\prime)$ is naturally isomorphic to
$\cl(V,Q)\hat{\otimes}\cl(V^\prime,Q^\prime)$.
\end{thm}
Let $\cl(V,Q)$ be the Clifford algebra over the field $\F=\R$, where
$V$ is a vector space endowed with quadratic form
$Q=x^2_1+\ldots+x^2_p-\ldots-x^2_{p+q}$. If $p+q$ is even and
$\omega^2=1$, then $\cl(V,Q)$ is called {\it positive} and
correspondingly {\it negative} if $\omega^2=-1$, that is,
$\cl_{p,q}>0$ if $p-q\equiv 0,4\pmod{8}$ and $\cl_{p,q}<0$ if
$p-q\equiv 2,6\pmod{8}$.
\begin{thm}[{\rm Karoubi \cite[Prop.~3.16]{Kar79}}]
1) If $\cl(V,Q)>0$ and $\dim V$ is even, then
\[
\cl(V\oplus V^{\p},Q\oplus
Q^{\p})\simeq\cl(V,Q)\otimes\cl(V^{\p},Q^{\p}).
\]
2) If $\cl(V,Q)<0$ and $\dim V$ is even, then
\[
\cl(V\oplus V^{\p},Q\oplus
Q^{\p})\simeq\cl(V,Q)\otimes\cl(V^{\p},-Q^{\p}).
\]
\end{thm}
\begin{sloppypar}\noindent
Over the field $\F=\C$ the Clifford algebra is always positive (if
$\omega^2=\e^2_{12\ldots p+q}=-1$, then we can suppose $\omega=
i\e_{12\ldots p+q}$). Thus, using Karoubi Theorem, we find that
\end{sloppypar}
\begin{equation}\label{3}
\underbrace{\C_2\otimes\cdots\otimes\C_2}_{m\,\text{times}}\simeq\C_{2m}.
\end{equation}
Therefore, the tensor product in (\ref{3}) is isomorphic to
$\C_{2m}$. For example, there are two different factorizations
$\cl_{1,1}\otimes\cl_{0,2}$ and $\cl_{1,1}\otimes\cl_{2,0}$ for the
spacetime algebra $\cl_{1,3}$ and Majorana algebra $\cl_{3,1}$.
\subsection{Supergroups} Let us consider graded tensor
products $\fA_1\hat{\otimes}\fA_2$ and
$\fA_1\hat{\otimes}\fA_2\hat{\otimes}\fA_3$ over the field $\F$ of
characteristics $0$, where $\fA=\cl_{p,q}$ over the field $\F=\R$
and $\fA=\C_n$ over $\F=\C$. These products possess the same
properties as the algebra $\fA$. Let $\e^\prime_i$
($i=1,\ldots,n_1$), $\e^{\prime\prime}_i$ ($i=1,\ldots,n_2$),
$\e^{\prime\prime\prime}_i$ ($i=1,\ldots,n_3$) be the units of the
algebras $\fA_i$. Then the units of the algebras
$\fA_1\hat{\otimes}\fA_2$ and
$\fA_1\hat{\otimes}\fA_2\hat{\otimes}\fA_3$ are the sets
$\e^\prime_i,\,\e^{\prime\prime}_j$ and
$\e^\prime_i,\,\e^{\prime\prime}_j,\,\e^{\prime\prime\prime}_k$,
respectively. General elements of the algebras
$\fA_1\hat{\otimes}\fA_2$ and
$\fA_1\hat{\otimes}\fA_2\hat{\otimes}\fA_3$ are formal polynomials
from $n_1+n_2$ and $n_1+n_2+n_3$ units. Dimensionalities of
$\fA_1\hat{\otimes}\fA_2$ and
$\fA_1\hat{\otimes}\fA_2\hat{\otimes}\fA_3$ are equal to
$2^{n_1+n_2}$ and $2^{n_1+n_2+n_3}$, respectively.

Let $d$ be a mapping of $\fA\hat{\otimes}\fA$ into $\fA$ ($d$
defines a multiplication operation in the algebra $\fA$:
$d_{a_1\otimes a_2}=a_1a_2$). Further, let us suppose
$\eta(k)=k\e_0$, where $\e_0$ is the unit of the algebra $\fA$,
$k\in\F$, and $\varepsilon_a=a(0)$. The mapping $\eta$ defines a
natural embedding of the field $\F$ into $\fA$, $\varepsilon$ is an
homomorphism of the algebra $\fA$ into the field $\F$.
\begin{defn}
The algebra $\fA$ is called a supergroup
space\index{space!supergroup} if there exist the homomorphism
$\gamma:\;\fA\rightarrow\fA\hat{\otimes}\fA$ and an isomorphism
$\bcirc:\;\fA\rightarrow\fA$ with the following commutative
diagrams:
\begin{gather}
\dgARROWLENGTH=2.5em
\begin{diagram}
\node[2]{\fA\hat{\otimes}\fA}\arrow{se,t}{I\otimes\gamma}\\
\node{\fA}\arrow{se,b}{\gamma}\arrow{ne,t}{\gamma}
\node[2]{\fA\hat{\otimes}\fA\hat{\otimes}\fA}\\
\node[2]{\fA\hat{\otimes}\fA}\arrow{ne,b}{\gamma\otimes I}
\end{diagram}\label{Diag1}\\[0.3cm]
\dgARROWLENGTH=2.5em
\begin{diagram}
\node[2]{\fA\hat{\otimes}\fA}\arrow{se,t}{I\otimes\varepsilon}\\
\node{\fA}\arrow{ne,t}{\gamma}\arrow[2]{<>,b}{I}\node[2]{\fA}
\end{diagram}\label{Diag2}\\[0.5cm]
\dgARROWLENGTH=2.5em
\begin{diagram}
\node{\fA\hat{\otimes}\fA}\arrow[2]{e,t}{I\otimes{\scriptstyle\bigcirc}}
\node[2]{\fA\hat{\otimes}\fA}\arrow{s,r}{d}\\
\node{\fA}\arrow{n,l}{\gamma}\arrow{e,b}{\varepsilon}\node{\F}
\arrow{e,b}{\eta}\node{\fA}
\end{diagram}\label{Diag3}
\end{gather}
\end{defn}
\begin{defn}
The totality $G(\fA,\gamma,\bcirc)$ is called a Cliffordian
supergroup.\index{supergroup}
\end{defn}
The homomorphism $\gamma$ is defined as follows. Let
\begin{eqnarray}
j:&&V\longrightarrow\cl(V,Q),\nonumber\\
j^\prime:&&V^\prime\longrightarrow\cl(V^\prime,Q^\prime)\nonumber
\end{eqnarray}
be canonical mappings (linear mappings of the vector spaces into the
algebras). Let us assume that $V$ and $V^\prime$ are mutually
orthogonal, then $V\oplus V^\prime$. Let
\[
j^{\prime\prime}:\; V\oplus V^\prime\longrightarrow
\cl(V,Q)\hat{\otimes}\cl(V^\prime,Q^\prime)
\]
be a canonical mapping of the direct sum $V\oplus V^\prime$ into the
graded tensor product of the Clifford algebras defined by the
formula $j^{\prime\prime}(v,v^\prime)=j(v)\otimes 1+1\otimes
j(v^\prime)$, where $v$ and $v^\prime$ are vectors of the spaces $V$
and $V^\prime$, respectively. Hence it immediately follows that
$(j^{\prime\prime}(v,v^\prime))^2=[(j(v))^2+(j^\prime(v^\prime))^2]\cdot
1= [Q(v)+Q^\prime(v^\prime)]\cdot 1$, since $v$ and $v^\prime$ are
orthogonal. Thus, the mapping $j^{\prime\prime}$ induces an
homomorphism
\[
\psi:\;\cl(V\oplus V^\prime,Q\oplus Q^\prime)\longrightarrow
\cl(V,Q)\hat{\otimes}\cl(V^\prime,Q^\prime).
\]
Let
\begin{eqnarray}
\gamma:&&\cl(V,Q)\longrightarrow\cl(V\oplus V^\prime,Q\oplus
Q^\prime),
\label{HomG}\\
\gamma^\prime:&&\cl(V^\prime,Q^\prime)\longrightarrow \cl(V\oplus
V^\prime,Q\oplus Q^\prime)\label{HomG'}
\end{eqnarray}
be homomorphisms induced by embeddings of $V$ and $V^\prime$ into
$V\oplus V^\prime$. Since $V$ and $V^\prime$ are orthogonal in
$V\oplus V^\prime$, then there exists an equality
$\gamma(x)\gamma^\prime(x^\prime)=(-1)^{\deg x\deg x^\prime}
\gamma^\prime(x^\prime)\gamma(x)$, where $x$ and $x^\prime$ are
homogeneous elements of the algebras $\cl(V,Q)$ and
$\cl(V^\prime,Q^\prime)$. Therefore, $\gamma$ and $\gamma^\prime$
define an homomorphism
\[
\theta:\;\cl(V,Q)\hat{\otimes}\cl(V^\prime,Q^\prime)\longrightarrow
\cl(V\oplus V^\prime,Q\oplus Q^\prime)
\]
by the formula $\theta(x\otimes
x^\prime)=\gamma(x)\gamma^\prime(x^\prime)$. It is obvious that the
homomorphisms $\psi$ and $\theta$ are mutually invertible.

The homomorphism $\gamma\otimes I$ is understood as follows. Let
$f(a,\e)=\sum\limits_{k\ge 0}\sum\limits_{i_1,\ldots,i_k}
a^{i_1\ldots i_k}\e_{i_1}\cdots\e_{i_k}$ be a general element of the
algebra $\fA$, $a^{i_1\ldots i_k}\in\F$. If
$f(a_1,\e_1;a_2,\e_2)=f_1(a_1,\e_1)f_2(a_2,\e_2)$ is an element of
the algebra $\fA\hat{\otimes}\fA$, then
\[
(\gamma\otimes I f)(a_1,\e_1;a_2,\e_2;a_3,\e_3)=(\gamma
f_1)(a_1,\e_1;a_2,\e_2) f_2(a_3,\e_3)
\]
defines an action of the homomorphism $\gamma\otimes I$. The
homomorphisms $I\otimes\gamma$, $I\otimes\varepsilon$,
$I\otimes\bcirc$ are defined analogously.

The diagrams (\ref{Diag1})--(\ref{Diag2}) present the axioms of
associativity, right unit and right inverse element, respectively.
The first two diagrams show that $\gamma$ and $\varepsilon$ define
on $\fA$ the structure of a co-algebra.

\begin{sloppypar}
The conditions (\ref{Diag1})--(\ref{Diag3}) coincide with axioms for
classical formal groups\index{group!classical formal} \cite{Boc46}.
The main difference from the classical case consists in the
permutation conditions $\e_i\e_j=(-1)^{\deg\e_i\deg\e_j}\e_j\e_i$
for the unuts of the graded algebra $\fA$. As in the case of
classical theory, one can establish a left analog of the equalities
(\ref{Diag2}) and (\ref{Diag3}) \cite{BK70}.\end{sloppypar}

Let us consider a supergroup\index{supergroup}
$G(\fA,\gamma,\bcirc)$ over the field $\F=\R$. In this case we have
$\fA=\cl_{p,q}$. The real Clifford algebra $\cl_{p,q}$ is central
simple if $p-q\not\equiv 1,5\pmod{8}$. The graded tensor product of
the two graded central simple algebras\index{algebra!graded central
simple} is also graded central simple \cite[Theorem 2]{Wal64}. It is
known that for the Clifford algebra with odd dimensionality, the
isomorphisms are as follows: $\cl^+_{p,q+1}\simeq\cl_{p,q}$ and
$\cl^+_{p+1,q}\simeq\cl_{q,p}$ \cite{Rash,Port}. Thus,
$\cl^+_{p,q+1}$ and $\cl^+_{p+1,q}$ are central simple algebras.
Further, in accordance with Chevalley Theorem for the graded tensor
product there is an isomorphism
$\cl_{p,q}\hat{\otimes}\cl_{p^{\p},q^{\p}}\simeq
\cl_{p+p^{\p},q+q^{\p}}$. Two algebras $\cl_{p,q}$ and
$\cl_{p^{\p},q^{\p}}$ are said to be of the same class if
$p+q^{\p}\equiv p^{\p}+q\pmod{8}$. The graded central simple
Clifford algebras over the field $\F=\R$ form eight similarity
classes, which, as it is easy to see, coincide with the eight types
of the algebras $\cl_{p,q}$. The set of these 8 types (classes)
forms a Brauer-Wall group $BW_{\R}$ \cite{Wal64,Lou81}. It is
obvious that an action of the supergroup
$BW_{\R}=G(\cl_{p,q},\gamma,\bcirc)$ has a cyclic structure, which
is formally equivalent to the action of cyclic
group\index{group!cyclic} $\dZ_8$. The cyclic structure of the
supergroup $BW_{\R}$ may be represented on the Trautman diagram
(spinorial clock) \cite{BTr87,BT88} (Fig. 1) by means of a
transition $\cl^+_{p,q}\stackrel{h}{\longrightarrow}\cl_{p,q}$ (the
round on the diagram is realized by an hour-hand). At this point,
the type of the algebra is defined on the diagram by an equality
$q-p=h+8r$, where $h\in\{0,\ldots,7\}$, $r\in\dZ$.

\begin{figure}[ht]
\[
\unitlength=0.5mm
\begin{picture}(100.00,110.00)

\put(97,67){$\C$}\put(105,64){$p-q\equiv 7\!\!\!\!\pmod{8}$}
\put(80,80){1} \put(75,93.3){$\cdot$} \put(75.5,93){$\cdot$}
\put(76,92.7){$\cdot$} \put(76.5,92.4){$\cdot$}
\put(77,92.08){$\cdot$} \put(77.5,91.76){$\cdot$}
\put(78,91.42){$\cdot$} \put(78.5,91.08){$\cdot$}
\put(79,90.73){$\cdot$} \put(79.5,90.37){$\cdot$}
\put(80,90.0){$\cdot$} \put(80.5,89.62){$\cdot$}
\put(81,89.23){$\cdot$} \put(81.5,88.83){$\cdot$}
\put(82,88.42){$\cdot$} \put(82.5,87.99){$\cdot$}
\put(83,87.56){$\cdot$} \put(83.5,87.12){$\cdot$}
\put(84,86.66){$\cdot$} \put(84.5,86.19){$\cdot$}
\put(85,85.70){$\cdot$} \put(85.5,85.21){$\cdot$}
\put(86,84.69){$\cdot$} \put(86.5,84.17){$\cdot$}
\put(87,83.63){$\cdot$} \put(87.5,83.07){$\cdot$}
\put(88,82.49){$\cdot$} \put(88.5,81.9){$\cdot$}
\put(89,81.29){$\cdot$} \put(89.5,80.65){$\cdot$}
\put(90,80){$\cdot$} \put(90.5,79.32){$\cdot$}
\put(91,78.62){$\cdot$} \put(91.5,77.89){$\cdot$}
\put(92,77.13){$\cdot$} \put(92.5,76.34){$\cdot$}
\put(93,75.51){$\cdot$} \put(93.5,74.65){$\cdot$}
\put(94,73.74){$\cdot$} \put(94.5,72.79){$\cdot$}
\put(96.5,73.74){\vector(1,-2){1}}
\put(80,20){3} \put(97,31){$\BH$}\put(105,28){$p-q\equiv
6\!\!\!\!\pmod{8}$} \put(75,6.7){$\cdot$} \put(75.5,7){$\cdot$}
\put(76,7.29){$\cdot$} \put(76.5,7.6){$\cdot$}
\put(77,7.91){$\cdot$} \put(77.5,8.24){$\cdot$}
\put(78,8.57){$\cdot$} \put(78.5,8.91){$\cdot$}
\put(79,9.27){$\cdot$} \put(79.5,9.63){$\cdot$} \put(80,10){$\cdot$}
\put(80.5,10.38){$\cdot$} \put(81,10.77){$\cdot$}
\put(81.5,11.17){$\cdot$} \put(82,11.58){$\cdot$}
\put(82.5,12.00){$\cdot$} \put(83,12.44){$\cdot$}
\put(83.5,12.88){$\cdot$} \put(84,13.34){$\cdot$}
\put(84.5,13.8){$\cdot$} \put(85,14.29){$\cdot$}
\put(85.5,14.79){$\cdot$} \put(86,15.3){$\cdot$}
\put(86.5,15.82){$\cdot$} \put(87,16.37){$\cdot$}
\put(87.5,16.92){$\cdot$} \put(88,17.5){$\cdot$}
\put(88.5,18.09){$\cdot$} \put(89,18.71){$\cdot$}
\put(89.5,19.34){$\cdot$} \put(90,20){$\cdot$}
\put(90.5,20.68){$\cdot$} \put(91,21.38){$\cdot$}
\put(91.5,22.11){$\cdot$} \put(92,22.87){$\cdot$}
\put(92.5,23.66){$\cdot$} \put(93,24.48){$\cdot$}
\put(93.5,25.34){$\cdot$} \put(94,26.25){$\cdot$}
\put(94.5,27.20){$\cdot$} \put(95,28.20){$\cdot$}
\put(20,80){7} \put(25,93.3){$\cdot$} \put(24.5,93){$\cdot$}
\put(24,92.7){$\cdot$} \put(23.5,92.49){$\cdot$}
\put(23,92.08){$\cdot$} \put(22.5,91.75){$\cdot$}
\put(22,91.42){$\cdot$} \put(21.5,91.08){$\cdot$}
\put(21,90.73){$\cdot$} \put(20.5,90.37){$\cdot$}
\put(20,90){$\cdot$} \put(19.5,89.62){$\cdot$}
\put(19,89.23){$\cdot$} \put(18.5,88.83){$\cdot$}
\put(18,88.42){$\cdot$} \put(17.5,87.99){$\cdot$}
\put(17,87.56){$\cdot$} \put(16.5,87.12){$\cdot$}
\put(16,86.66){$\cdot$} \put(15.5,86.19){$\cdot$}
\put(15,85.70){$\cdot$} \put(14.5,85.21){$\cdot$}
\put(14,84.69){$\cdot$} \put(13.5,84.17){$\cdot$}
\put(13,83.63){$\cdot$} \put(12.5,83.07){$\cdot$}
\put(12,82.49){$\cdot$} \put(11.5,81.9){$\cdot$}
\put(11,81.29){$\cdot$} \put(10.5,80.65){$\cdot$}
\put(10,80){$\cdot$} \put(9.5,79.32){$\cdot$} \put(9,78.62){$\cdot$}
\put(8.5,77.89){$\cdot$} \put(8,77.13){$\cdot$}
\put(7.5,76.34){$\cdot$} \put(7,75.51){$\cdot$}
\put(6.5,74.65){$\cdot$} \put(6,73.79){$\cdot$}
\put(5.5,72.79){$\cdot$} \put(5,71.79){$\cdot$}
\put(20,20){5} \put(25,6.7){$\cdot$} \put(24.5,7){$\cdot$}
\put(24,7.29){$\cdot$} \put(23.5,7.6){$\cdot$}
\put(23,7.91){$\cdot$} \put(22.5,8.24){$\cdot$}
\put(22,8.57){$\cdot$} \put(21.5,8.91){$\cdot$}
\put(21,9.27){$\cdot$} \put(20.5,9.63){$\cdot$} \put(20,10){$\cdot$}
\put(19.5,10.38){$\cdot$} \put(19,10.77){$\cdot$}
\put(18.5,11.17){$\cdot$} \put(18,11.58){$\cdot$}
\put(17.5,12){$\cdot$} \put(17,12.44){$\cdot$}
\put(16.5,12.88){$\cdot$} \put(16,13.34){$\cdot$}
\put(15.5,13.8){$\cdot$} \put(15,14.29){$\cdot$}
\put(14.5,14.79){$\cdot$} \put(14,15.3){$\cdot$}
\put(13.5,15.82){$\cdot$} \put(13,16.37){$\cdot$}
\put(12.5,16.92){$\cdot$} \put(12,17.5){$\cdot$}
\put(11.5,18.09){$\cdot$} \put(11,18.71){$\cdot$}
\put(10.5,19.34){$\cdot$} \put(10,20){$\cdot$}
\put(9.5,20.68){$\cdot$} \put(9,21.38){$\cdot$}
\put(8.5,22.11){$\cdot$} \put(8,22.87){$\cdot$}
\put(7.5,23.66){$\cdot$} \put(7,24.48){$\cdot$}
\put(6.5,25.34){$\cdot$} \put(6,26.25){$\cdot$}
\put(5.5,27.20){$\cdot$} \put(5,28.20){$\cdot$}
\put(13,97){$\R\oplus\R$}\put(-55,105){$p-q\equiv
1\!\!\!\!\pmod{8}$} \put(50,93){0} \put(50,100){$\cdot$}
\put(49.5,99.99){$\cdot$} \put(49,99.98){$\cdot$}
\put(48.5,99.97){$\cdot$} \put(48,99.96){$\cdot$}
\put(47.5,99.94){$\cdot$} \put(47,99.91){$\cdot$}
\put(46.5,99.86){$\cdot$} \put(46,99.84){$\cdot$}
\put(45.5,99.8){$\cdot$} \put(45,99.75){$\cdot$}
\put(44.5,99.7){$\cdot$} \put(44,99.64){$\cdot$}
\put(43.5,99.57){$\cdot$} \put(43,99.51){$\cdot$}
\put(42.5,99.43){$\cdot$} \put(42,99.35){$\cdot$}
\put(41.5,99.27){$\cdot$} \put(41,99.18){$\cdot$}
\put(40.5,99.09){$\cdot$} \put(40,98.99){$\cdot$}
\put(39.5,98.88){$\cdot$} \put(39,98.77){$\cdot$}
\put(38.5,98.66){$\cdot$} \put(38,98.54){$\cdot$}
\put(37.5,98.41){$\cdot$} \put(37,98.28){$\cdot$}
\put(50.5,99.99){$\cdot$} \put(51,99.98){$\cdot$}
\put(51.5,99.97){$\cdot$} \put(52,99.96){$\cdot$}
\put(52.5,99.94){$\cdot$} \put(53,99.91){$\cdot$}
\put(53.5,99.86){$\cdot$} \put(54,99.84){$\cdot$}
\put(54.5,99.8){$\cdot$} \put(55,99.75){$\cdot$}
\put(55.5,99.7){$\cdot$} \put(56,99.64){$\cdot$}
\put(56.5,99.57){$\cdot$} \put(57,99.51){$\cdot$}
\put(57.5,99.43){$\cdot$} \put(58,99.35){$\cdot$}
\put(58.5,99.27){$\cdot$} \put(59,99.18){$\cdot$}
\put(59.5,99.09){$\cdot$} \put(60,98.99){$\cdot$}
\put(60.5,98.88){$\cdot$} \put(61,98.77){$\cdot$}
\put(61.5,98.66){$\cdot$} \put(62,98.54){$\cdot$}
\put(62.5,98.41){$\cdot$} \put(63,98.28){$\cdot$}
\put(68,97){$\R$}\put(73,105){$p-q\equiv 0\!\!\!\!\pmod{8}$}
\put(50,7){4} \put(68,2){$\BH\oplus\BH$}\put(90,-4){$p-q\equiv
5\!\!\!\!\pmod{8}$} \put(50,0){$\cdot$} \put(50.5,0){$\cdot$}
\put(51,0.01){$\cdot$} \put(51.5,0.02){$\cdot$}
\put(52,0.04){$\cdot$} \put(52.5,0.06){$\cdot$}
\put(53,0.09){$\cdot$} \put(53.5,0.12){$\cdot$}
\put(54,0.16){$\cdot$} \put(54.5,0.2){$\cdot$}
\put(55,0.25){$\cdot$} \put(55.5,0.3){$\cdot$}
\put(56,0.36){$\cdot$} \put(56.5,0.42){$\cdot$}
\put(57,0.49){$\cdot$} \put(57.5,0.56){$\cdot$}
\put(58,0.64){$\cdot$} \put(58.5,0.73){$\cdot$}
\put(59,0.82){$\cdot$} \put(59.5,0.91){$\cdot$}
\put(60,1.01){$\cdot$} \put(60.5,1.11){$\cdot$}
\put(61,1.22){$\cdot$} \put(61.5,1.34){$\cdot$}
\put(62,1.46){$\cdot$} \put(62.5,1.59){$\cdot$}
\put(63,1.72){$\cdot$} \put(49.5,0){$\cdot$} \put(49,0.01){$\cdot$}
\put(48.5,0.02){$\cdot$} \put(48,0.04){$\cdot$}
\put(47.5,0.06){$\cdot$} \put(47,0.09){$\cdot$}
\put(46.5,0.12){$\cdot$} \put(46,0.16){$\cdot$}
\put(45.5,0.2){$\cdot$} \put(45,0.25){$\cdot$}
\put(44.5,0.3){$\cdot$} \put(44,0.36){$\cdot$}
\put(43.5,0.42){$\cdot$} \put(43,0.49){$\cdot$}
\put(42.5,0.56){$\cdot$} \put(42,0.64){$\cdot$}
\put(41.5,0.73){$\cdot$} \put(41,0.82){$\cdot$}
\put(40.5,0.91){$\cdot$} \put(40,1.01){$\cdot$}
\put(39.5,1.11){$\cdot$} \put(39,1.22){$\cdot$}
\put(38.5,1.34){$\cdot$} \put(38,1.46){$\cdot$}
\put(37.5,1.59){$\cdot$} \put(37,1.72){$\cdot$}
\put(28,3){$\BH$}\put(-40,-4){$p-q\equiv 4\!\!\!\!\pmod{8}$}
\put(93,50){2} \put(98.28,63){$\cdot$} \put(98.41,62.5){$\cdot$}
\put(98.54,62){$\cdot$} \put(98.66,61.5){$\cdot$}
\put(98.77,61){$\cdot$} \put(98.88,60.5){$\cdot$}
\put(98.99,60){$\cdot$} \put(99.09,59.5){$\cdot$}
\put(99.18,59){$\cdot$} \put(99.27,58.5){$\cdot$}
\put(99.35,58){$\cdot$} \put(99.43,57.5){$\cdot$}
\put(99.51,57){$\cdot$} \put(99.57,56.5){$\cdot$}
\put(99.64,56){$\cdot$} \put(99.7,55.5){$\cdot$}
\put(99.75,55){$\cdot$} \put(99.8,54.5){$\cdot$}
\put(99.84,54){$\cdot$} \put(99.86,53.5){$\cdot$}
\put(99.91,53){$\cdot$} \put(99.94,52.5){$\cdot$}
\put(99.96,52){$\cdot$} \put(99.97,51.5){$\cdot$}
\put(99.98,51){$\cdot$} \put(99.99,50.5){$\cdot$}
\put(100,50){$\cdot$} \put(98.28,37){$\cdot$}
\put(98.41,37.5){$\cdot$} \put(98.54,38){$\cdot$}
\put(98.66,38.5){$\cdot$} \put(98.77,39){$\cdot$}
\put(98.88,39.5){$\cdot$} \put(98.99,40){$\cdot$}
\put(99.09,40.5){$\cdot$} \put(99.18,41){$\cdot$}
\put(99.27,41.5){$\cdot$} \put(99.35,42){$\cdot$}
\put(99.43,42.5){$\cdot$} \put(99.51,43){$\cdot$}
\put(99.57,43.5){$\cdot$} \put(99.64,44){$\cdot$}
\put(99.7,44.5){$\cdot$} \put(99.75,45){$\cdot$}
\put(99.8,45.5){$\cdot$} \put(99.84,46){$\cdot$}
\put(99.86,46.5){$\cdot$} \put(99.91,47){$\cdot$}
\put(99.94,47.5){$\cdot$} \put(99.96,48){$\cdot$}
\put(99.97,48.5){$\cdot$} \put(99.98,49){$\cdot$}
\put(99.99,49.5){$\cdot$}
\put(7,50){6} \put(1,32){$\C$}\put(-67,29){$p-q\equiv
3\!\!\!\!\pmod{8}$} \put(1.72,63){$\cdot$} \put(1.59,62.5){$\cdot$}
\put(1.46,62){$\cdot$} \put(1.34,61.5){$\cdot$}
\put(1.22,61){$\cdot$} \put(1.11,60.5){$\cdot$}
\put(1.01,60){$\cdot$} \put(0.99,59.5){$\cdot$}
\put(0.82,59){$\cdot$} \put(0.73,58.5){$\cdot$}
\put(0.64,58){$\cdot$} \put(0.56,57.5){$\cdot$}
\put(0.49,57){$\cdot$} \put(0.42,56.5){$\cdot$}
\put(0.36,56){$\cdot$} \put(0.3,55.5){$\cdot$}
\put(0.25,55){$\cdot$} \put(0.2,54.5){$\cdot$}
\put(0.16,54){$\cdot$} \put(0.12,53.5){$\cdot$}
\put(0.09,53){$\cdot$} \put(0.06,52.5){$\cdot$}
\put(0.04,52){$\cdot$} \put(0.02,51.5){$\cdot$}
\put(0.01,51){$\cdot$} \put(0,50.5){$\cdot$} \put(0,50){$\cdot$}
\put(1.72,37){$\cdot$} \put(1.59,37.5){$\cdot$}
\put(1.46,38){$\cdot$} \put(1.34,38.5){$\cdot$}
\put(1.22,39){$\cdot$} \put(1.11,39.5){$\cdot$}
\put(1.01,40){$\cdot$} \put(0.99,40.5){$\cdot$}
\put(0.82,41){$\cdot$} \put(0.73,41.5){$\cdot$}
\put(0.64,42){$\cdot$} \put(0.56,42.5){$\cdot$}
\put(0.49,43){$\cdot$} \put(0.42,43.5){$\cdot$}
\put(0.36,44){$\cdot$} \put(0.3,44.5){$\cdot$}
\put(0.25,45){$\cdot$} \put(0.2,45.5){$\cdot$}
\put(0.16,46){$\cdot$} \put(0.12,46.5){$\cdot$}
\put(0.09,47){$\cdot$} \put(0.06,47.5){$\cdot$}
\put(0.04,48){$\cdot$} \put(0.02,48.5){$\cdot$}
\put(0.01,49){$\cdot$} \put(0,49.5){$\cdot$}
\put(0.5,67){$\R$}\put(-67,75){$p-q\equiv 2\!\!\!\!\pmod{8}$}
\end{picture}
\]
\vspace{2ex}
\begin{center}
\begin{minipage}{25pc}{\small
{\bf Fig.1} The Trautman diagram for the Brauer-Wall group
$BW_{\R}$.}
\end{minipage}
\end{center}
\end{figure}
\medskip
It is obvious that a group structure over $\cl_{p,q}$, defined by
$BW_{\R}$, is related with the Atiyah-Bott-Shapiro periodicity
\cite{AtBSh}. In accordance with \cite{AtBSh}, the Clifford algebra
over the field $\F=\R$ is modulo 8 periodic\index{periodicity!modulo
8}:
$\cl_{p+8,q}\simeq\cl_{p,q}\otimes\cl_{8,0}\,(\cl_{p,q+8}\simeq\cl_{p,q}
\otimes\cl_{0,8})$.

Further, over the field $\F=\C$, there exist two types of the
complex Clifford algebras: $\C_n$ and
$\C_{n+1}\simeq\C_n\oplus\C_n$. Therefore, an action of $BW_{\C}$,
defined on the set of these two types, is formally equivalent to the
action of cyclic group $\dZ_2$. The cyclic structure of the
supergroup $BW_{\C}=G(\C_n,\gamma,\bcirc)$ may be represented on the
following Trautman diagram (Fig. 2) by means of a transition
$\C^+_n\stackrel{h}{\longrightarrow}\C_n$ (the round on the diagram
is realized by an hour-hand). At this point, the type of the algebra
on the diagram is defined by an equality $n=h+2r$, where
$h\in\{0,1\}$, $r\in\dZ$.
\begin{figure}[t]
\[
\unitlength=0.5mm
\begin{picture}(50.00,60.00)(0,0)
\put(5,25){1} \put(42,25){0} \put(22,-4){$\C_{{\rm even}}$}
\put(22,55){$\C_{{\rm odd}}$} \put(3,-13){$n\equiv
0\!\!\!\!\pmod{2}$} \put(3,64){$n\equiv 1\!\!\!\!\pmod{2}$}
\put(20,49.49){$\cdot$} \put(19.5,49.39){$\cdot$}
\put(19,49.27){$\cdot$} \put(18.5,49.14){$\cdot$}
\put(18,49){$\cdot$} \put(17.5,48.85){$\cdot$}
\put(17,48.68){$\cdot$} \put(16.5,48.51){$\cdot$}
\put(16,48.32){$\cdot$} \put(15.5,48.12){$\cdot$}
\put(15,47.91){$\cdot$} \put(14.5,47.69){$\cdot$}
\put(14,47.45){$\cdot$} \put(13.5,47.2){$\cdot$}
\put(13,46.93){$\cdot$} \put(12.5,46.65){$\cdot$}
\put(12,46.35){$\cdot$} \put(11.5,46.04){$\cdot$}
\put(11,45.71){$\cdot$} \put(10.5,45.36){$\cdot$}
\put(10,45){$\cdot$} \put(9.5,44.61){$\cdot$} \put(9,44.21){$\cdot$}
\put(8.5,43.78){$\cdot$} \put(8,43.33){$\cdot$}
\put(7.5,42.85){$\cdot$} \put(7,42.35){$\cdot$}
\put(6.5,41.81){$\cdot$} \put(6,41.25){$\cdot$}
\put(5.5,40.64){$\cdot$} \put(5,40){$\cdot$} \put(4.5,39.3){$\cdot$}
\put(4,38.56){$\cdot$} \put(3.5,37.76){$\cdot$}
\put(3,36.87){$\cdot$} \put(2.5,35.89){$\cdot$}
\put(2,34.79){$\cdot$} \put(1.5,33.53){$\cdot$} \put(1,32){$\cdot$}
\put(0.5,29.97){$\cdot$}
\put(30,49.49){$\cdot$} \put(30.5,49.39){$\cdot$}
\put(31,49.27){$\cdot$} \put(31.5,49.14){$\cdot$}
\put(32,49){$\cdot$} \put(32.5,48.85){$\cdot$}
\put(33,48.68){$\cdot$} \put(33.5,48.51){$\cdot$}
\put(34,48.32){$\cdot$} \put(34.5,48.12){$\cdot$}
\put(35,47.91){$\cdot$} \put(35.5,47.69){$\cdot$}
\put(36,47.45){$\cdot$} \put(36.5,47.2){$\cdot$}
\put(37,46.93){$\cdot$} \put(37.5,46.65){$\cdot$}
\put(38,46.35){$\cdot$} \put(38.5,46.04){$\cdot$}
\put(39,45.71){$\cdot$} \put(39.5,45.36){$\cdot$}
\put(40,45){$\cdot$} \put(40.5,44.61){$\cdot$}
\put(41,44.21){$\cdot$} \put(41.5,43.78){$\cdot$}
\put(42,43.33){$\cdot$} \put(42.5,42.85){$\cdot$}
\put(43,42.35){$\cdot$} \put(43.5,41.81){$\cdot$}
\put(44,41.25){$\cdot$} \put(44.5,40.64){$\cdot$}
\put(45,40){$\cdot$} \put(45.5,39.3){$\cdot$}
\put(46,38.56){$\cdot$} \put(46.5,37.76){$\cdot$}
\put(47,36.87){$\cdot$} \put(47.5,35.89){$\cdot$}
\put(48,34.79){$\cdot$} \put(48.5,33.53){$\cdot$}
\put(49,32){$\cdot$} \put(49.5,29.97){$\cdot$}
\put(0,25){$\cdot$} \put(0,24.5){$\cdot$} \put(0.02,24){$\cdot$}
\put(0.04,23.5){$\cdot$} \put(0.08,23){$\cdot$}
\put(0.12,22.5){$\cdot$} \put(0.18,22){$\cdot$}
\put(0.25,21.5){$\cdot$} \put(0.32,21){$\cdot$}
\put(0.4,20.5){$\cdot$} \put(0.5,20){$\cdot$}
\put(0.61,19.5){$\cdot$} \put(0.73,19){$\cdot$}
\put(0.85,18.5){$\cdot$} \put(1,18){$\cdot$}
\put(1.15,17.5){$\cdot$} \put(1.31,17){$\cdot$}
\put(1.49,16.5){$\cdot$} \put(1.68,16){$\cdot$}
\put(1.88,15.5){$\cdot$} \put(2.09,15){$\cdot$}
\put(2.31,14.5){$\cdot$} \put(2.55,14){$\cdot$}
\put(2.8,13.5){$\cdot$} \put(3.06,13){$\cdot$} \put(0,25.5){$\cdot$}
\put(0.02,26){$\cdot$} \put(0.04,26.5){$\cdot$}
\put(0.08,27){$\cdot$} \put(0.12,27.5){$\cdot$}
\put(0.18,28){$\cdot$} \put(0.25,28.5){$\cdot$}
\put(0.32,29){$\cdot$} \put(0.4,29.5){$\cdot$} \put(0.5,30){$\cdot$}
\put(0.61,30.5){$\cdot$} \put(0.73,31){$\cdot$}
\put(0.85,31.5){$\cdot$} \put(1,32){$\cdot$}
\put(1.15,32.5){$\cdot$} \put(1.31,33){$\cdot$}
\put(1.49,33.5){$\cdot$} \put(1.68,34){$\cdot$}
\put(1.88,34.5){$\cdot$} \put(2.09,35){$\cdot$}
\put(2.31,35.5){$\cdot$} \put(2.55,36){$\cdot$}
\put(2.8,36.5){$\cdot$} \put(3.06,37){$\cdot$}
\put(50,25){$\cdot$} \put(49.99,24.5){$\cdot$}
\put(49.98,24){$\cdot$} \put(49.95,23.5){$\cdot$}
\put(49.92,23){$\cdot$} \put(49.87,22.5){$\cdot$}
\put(49.82,22){$\cdot$} \put(49.75,21.5){$\cdot$}
\put(49.68,21){$\cdot$} \put(49.51,20.5){$\cdot$}
\put(49.49,20){$\cdot$} \put(49.39,19.5){$\cdot$}
\put(49.27,19){$\cdot$} \put(49.14,18.5){$\cdot$}
\put(49,18){$\cdot$} \put(48.85,17.5){$\cdot$}
\put(48.69,17){$\cdot$} \put(48.51,16.5){$\cdot$}
\put(48.32,16){$\cdot$} \put(48.12,15.5){$\cdot$}
\put(47.91,15){$\cdot$} \put(47.69,14.5){$\cdot$}
\put(47.45,14){$\cdot$} \put(47.2,13.5){$\cdot$}
\put(46.93,13){$\cdot$} \put(50,25){$\cdot$}
\put(49.99,25.5){$\cdot$} \put(49.98,26){$\cdot$}
\put(49.95,26.5){$\cdot$} \put(49.92,27){$\cdot$}
\put(49.87,27.5){$\cdot$} \put(49.82,28){$\cdot$}
\put(49.75,28.5){$\cdot$} \put(49.68,29){$\cdot$}
\put(49.51,29.5){$\cdot$} \put(49.49,30){$\cdot$}
\put(49.39,30.5){$\cdot$} \put(49.27,31){$\cdot$}
\put(49.14,31.5){$\cdot$} \put(49,32){$\cdot$}
\put(48.85,32.5){$\cdot$} \put(48.69,33){$\cdot$}
\put(48.51,33.5){$\cdot$} \put(48.32,34){$\cdot$}
\put(48.12,34.5){$\cdot$} \put(47.91,35){$\cdot$}
\put(47.69,35.5){$\cdot$} \put(47.45,36){$\cdot$}
\put(47.2,36.5){$\cdot$} \put(46.93,37){$\cdot$}
\put(20,0.5){$\cdot$} \put(19.5,0.61){$\cdot$}
\put(19,0.73){$\cdot$} \put(18.5,0.86){$\cdot$} \put(18,1){$\cdot$}
\put(17.5,1.15){$\cdot$} \put(17,1.31){$\cdot$}
\put(16.5,1.49){$\cdot$} \put(16,1.68){$\cdot$}
\put(15.5,1.87){$\cdot$} \put(15,2.09){$\cdot$}
\put(14.5,2.31){$\cdot$} \put(14,2.55){$\cdot$}
\put(13.5,2.8){$\cdot$} \put(13,3.06){$\cdot$}
\put(12.5,3.35){$\cdot$} \put(12,3.64){$\cdot$}
\put(11.5,3.96){$\cdot$} \put(11,4.29){$\cdot$}
\put(10.5,4.63){$\cdot$} \put(10,5){$\cdot$} \put(9.5,5.38){$\cdot$}
\put(9,5.79){$\cdot$} \put(8.5,6.22){$\cdot$} \put(8,6.67){$\cdot$}
\put(7.5,7.15){$\cdot$} \put(7,7.65){$\cdot$}
\put(6.5,8.18){$\cdot$} \put(6,8.75){$\cdot$}
\put(5.5,9.35){$\cdot$} \put(5,10){$\cdot$} \put(4.5,10.69){$\cdot$}
\put(4,11.43){$\cdot$} \put(3.5,12.24){$\cdot$}
\put(3,13.12){$\cdot$} \put(2.5,14.10){$\cdot$}
\put(2,15.20){$\cdot$} \put(1.5,16.47){$\cdot$} \put(1,18){$\cdot$}
\put(0.5,20.02){$\cdot$}
\put(30,0.5){$\cdot$} \put(30.5,0.61){$\cdot$}
\put(31,0.73){$\cdot$} \put(31.5,0.86){$\cdot$} \put(32,1){$\cdot$}
\put(32.5,1.15){$\cdot$} \put(33,1.31){$\cdot$}
\put(33.5,1.49){$\cdot$} \put(34,1.68){$\cdot$}
\put(34.5,1.87){$\cdot$} \put(35,2.09){$\cdot$}
\put(35.5,2.31){$\cdot$} \put(36,2.55){$\cdot$}
\put(36.5,2.8){$\cdot$} \put(37,3.06){$\cdot$}
\put(37.5,3.35){$\cdot$} \put(38,3.64){$\cdot$}
\put(38.5,3.96){$\cdot$} \put(39,4.29){$\cdot$}
\put(39.5,4.63){$\cdot$} \put(40,5){$\cdot$}
\put(40.5,5.38){$\cdot$} \put(41,5.79){$\cdot$}
\put(41.5,6.22){$\cdot$} \put(42,6.67){$\cdot$}
\put(42.5,7.15){$\cdot$} \put(43,7.65){$\cdot$}
\put(43.5,8.18){$\cdot$} \put(44,8.75){$\cdot$}
\put(44.5,9.35){$\cdot$} \put(45,10){$\cdot$}
\put(45.5,10.69){$\cdot$} \put(46,11.43){$\cdot$}
\put(46.5,12.24){$\cdot$} \put(47,13.12){$\cdot$}
\put(47.5,14.10){$\cdot$} \put(48,15.20){$\cdot$}
\put(48.5,16.47){$\cdot$} \put(49,18){$\cdot$}
\put(49.5,20.02){$\cdot$}

\end{picture}
\]
\vspace{2ex}
\begin{center}
\begin{minipage}{25pc}{\small
{\bf Fig.2} The Trautman diagram for the Brauer-Wall group
$BW_{\C}$.}
\end{minipage}
\end{center}
\end{figure}
\medskip
It is obvious that a group structure over $\C_n$, defined by the
group $BW_{\C}$, is related immediately with a modulo 2 periodicity
of the complex Clifford algebras \cite{AtBSh,Kar79}:
$\C_{n+2}\simeq\C_n\otimes\C_2$.
\subsection{Superalgebras}
As in the case of classical theory of groups and algebras, the
supergroup $G(\fA,\gamma,\bcirc)$ corresponds to a superalgebra (or
$\dZ_2$-graded Lie algebra). $\dZ_2$-graded linear space over the
field $\F$, endowed with a bilinear operation $\ld X_1,X_2\rd$, is
called a \emph{superalgebra}. At this point, the operation $\ld
X_1,X_2\rd$ satisfies the following conditions:
\begin{gather}
\ld X_1,X_2\rd=-(-1)^{\deg(X_1)\deg(X_2)}\ld X_2,X_1\rd,\nonumber\\
(-1)^{\deg(X_1)\deg(X_3)}\ld X_1,\ld X_2,X_3\rd\rd+
(-1)^{\deg(X_2)\deg(X_1)}\ld X_2,\ld X_3,X_1\rd\rd+\nonumber\\
+(-1)^{\deg(X_3)\deg(X_2)}\ld X_3,\ld X_1,X_2\rd\rd=0,\nonumber\\
\deg\left(\ld X_1,X_2\rd\right)=\deg(X_1)+\deg(X_2).\nonumber
\end{gather}
\section{Clifford algebras and representations of $\spin_+(1,3)$}
Let us consider the field
\begin{equation}\label{FieldL}
\boldsymbol{\psi}(\balpha)=\langle x,\fg\,|\boldsymbol{\psi}\rangle,
\end{equation}
where $x\in T_4$, $\fg\in\spin_+(1,3)$. The spinor group
$\spin_+(1,3)\simeq\SU(2)\otimes\SU(2)$ is a universal covering of
the proper orthochronous Lorentz group $\SO_0(1,3)$. The parameters
$x\in T_4$ and $\fg\in\spin_+(1,3)$ describe position and
orientation of the extended object defined by the field
(\ref{FieldL}) (the field on the Poincar\'{e} group).
\begin{thm}\label{tinf}
\begin{sloppypar}\noindent
Let $\spin_+(1,3)$ be the universal covering of the proper
orthochronous Lorentz group $\SO_0(1,3)$. Then over $\F=\C$ the
field $(l,0)\oplus(0,\dot{l})$ is constructed in the framework of a
complex finite-dimensional representation
$\fC^{l_0+l_1-1,0}\oplus\fC^{0,l_0-l_1+1}$ of $\spin_+(1,3)$ defined
on the spinspace $\dS_{2^k}\otimes\dS_{2^r}$ with the algebra
\end{sloppypar}
\[
\underbrace{\C_2\otimes\C_2\otimes\cdots\otimes\C_2}_{k\;\text{times}}\bigoplus
\underbrace{\overset{\ast}{\C}_2\otimes\overset{\ast}{\C}_2\otimes\cdots\otimes
\overset{\ast}{\C}_2}_{r\;\text{times}},
\]
where $(l_0,l_1)=\left(\frac{k}{2},\frac{k}{2}+1\right)$,
$(-l_0,l_1)= \left(-\frac{r}{2},\frac{r}{2}+1\right)$. In turn, the
field
$(l^\prime,l^{\prime\prime})\oplus(\dot{l}^{\prime\prime},\dot{l}^\prime)$
is constructed in the framework of a representation
$\fC^{l_0+l_1-1,l_0-l_1+1}\oplus\fC^{l_0-l_1+1,l_0+l_1-1}$ of
$\spin_+(1,3)$ defined on the spinspace
$\dS_{2^{k+r}}\oplus\dS_{2^{k+r}}$ with the algebra
\[
\underbrace{\C_2\otimes\C_2\otimes\cdots\otimes\C_2\bigotimes
\overset{\ast}{\C}_2\otimes\overset{\ast}{\C}_2\otimes\cdots\otimes
\overset{\ast}{\C}_2}_{k+r\;\text{times}}\bigoplus
\underbrace{\overset{\ast}{\C}_2\otimes\overset{\ast}{\C}_2\otimes\cdots\otimes
\overset{\ast}{\C}_2\bigotimes\C_2\otimes\C_2\otimes\cdots\otimes\C_2}_{r+k\;\text{times}},
\]
where $(l_0,l_1)=\left(\frac{k-r}{2}, \frac{k+r}{2}+1\right)$. In
the case of $\F=\R$ the field $\boldsymbol{\psi}(\balpha)$ of type
$(l_0,0)$ is constructed within real representations
$\fM^+=\{\fR^{l_0}_0,\fR^{l_0}_2,\fH^{l_0}_4,\fH^{l_0}_6\}$ of
$\spin_+(1,3)$ defined on the spinspace $\dS_{2^r}$ with the algebra
\[
\cl_{p,q}\simeq\underbrace{\cl_{s_i,t_j}\otimes\cl_{s_i,t_j}\otimes\cdots
\otimes\cl_{s_i,t_j}}_{r\;\text{times}},
\]
where $s_i,t_j\in\{0,1,2\}$. In turn, the field
$\boldsymbol{\psi}(\balpha)$ of type $(l_0,0)\cup(0,l_0)$ is
constructed within representations
$\fM^-=\{\fC^{l_0}_{3,7},\fR^{l_0}_{0,2}\cup\fR^{l_0}_{0,2},\fH^{l_0}_{4,6}\cup\fH^{l_0}_{4,6}\}$
of $\spin_+(1,3)$ defined on the double spinspace
$\dS_{2^r}\cup\dS_{2^r}$ with the algebra
\[
\cl_{p,q}\simeq\underbrace{\cl_{s_i,t_j}\otimes\cl_{s_i,t_j}\otimes\cdots
\otimes\cl_{s_i,t_j}}_{r\;\text{times}}\bigcup\underbrace{\cl_{s_i,t_j}\otimes\cl_{s_i,t_j}\otimes\cdots
\otimes\cl_{s_i,t_j}}_{r\;\text{times}}.
\]
\end{thm}
\begin{proof}1) Complex representations.\\
Since the spacetime algebra $\cl_{1,3}$ is the simple algebra, then
the map (\ref{Simple}) gives an irreducible representation of
$\cl_{1,3}$ in the spinspace $\dS_2(\BH)$. In turn, representations
of the group $\spin_+(1,3)\in\cl^+_{1,3}\simeq\cl_{3,0}$ are defined
in the spinspace $\dS_2(\C)$.
\begin{sloppypar} Let us consider now spintensor representations of the
group $\fG_+\simeq\SL(2;\C)$ which, as is known, form the base of
all the finite-dimensional representations of the Lorentz group, and
also we consider their relationship with the complex Clifford
algebras. From each complex Clifford algebra
$\C_n=\C\otimes\cl_{p,q}\; (n=p+q)$ we obtain the spinspace
$\dS_{2^{n/2}}$ which is a complexification of the minimal left
ideal of the algebra $\cl_{p,q}$: $\dS_{2^{n/2}}=\C\otimes
I_{p,q}=\C\otimes\cl_{p,q} f_{pq}$, where $f_{pq}$ is the primitive
idempotent of the algebra $\cl_{p,q}$. Further, a spinspace related
with the Pauli algebra $\C_2$ has the form $\dS_2=\C\otimes
I_{2,0}=\C\otimes\cl_{2,0}f_{20}$ or $\dS_2=\C\otimes
I_{1,1}=\C\otimes\cl_{1,1}f_{11}(\C\otimes I_{0,2}=
\C\otimes\cl_{0,2}f_{02})$. Therefore, the tensor product of the $k$
algebras $\C_2$ induces a tensor product of the $k$ spinspaces
$\dS_2$:\end{sloppypar}
\[
\dS_2\otimes\dS_2\otimes\cdots\otimes\dS_2=\dS_{2^k}.
\]
Vectors of the spinspace $\dS_{2^k}$ (or elements of the minimal
left ideal of $\C_{2k}$) are spintensors of the following form:
\begin{equation}\label{6.16}
\boldsymbol{s}^{\alpha_1\alpha_2\cdots\alpha_k}=\sum
\boldsymbol{s}^{\alpha_1}\otimes
\boldsymbol{s}^{\alpha_2}\otimes\cdots\otimes
\boldsymbol{s}^{\alpha_k},
\end{equation}
where summation is produced on all the index collections
$(\alpha_1\ldots\alpha_k)$, $\alpha_i=1,2$. For the each spinor
$\boldsymbol{s}^{\alpha_i}$ from (\ref{6.16}) we have ${}^\prime
\boldsymbol{s}^{\alpha^\prime_i}=
\sigma^{\alpha^\prime_i}_{\alpha_i}\boldsymbol{s}^{\alpha_i}$.
Therefore, in general case we obtain
\begin{equation}\label{6.17}
{}^\prime
\boldsymbol{s}^{\alpha^\prime_1\alpha^\prime_2\cdots\alpha^\prime_k}=\sum
\sigma^{\alpha^\prime_1}_{\alpha_1}\sigma^{\alpha^\prime_2}_{\alpha_2}\cdots
\sigma^{\alpha^\prime_k}_{\alpha_k}\boldsymbol{s}^{\alpha_1\alpha_2\cdots\alpha_k}.
\end{equation}
A representation (\ref{6.17}) is called {\it undotted spintensor
representation of the proper Lorentz group of the rank $k$}.

Further, let $\overset{\ast}{\C}_2$ be the Pauli algebra with the
coefficients which are complex conjugate to the coefficients of
$\C_2$. Let us show that the algebra $\overset{\ast}{\C}_2$ is
derived from $\C_2$ under action of the automorphism
$\cA\rightarrow\cA^\star$ or antiautomorphism
$\cA\rightarrow\widetilde{\cA}$ (see the definitions (\ref{auto16})
and (\ref{auto19})). Indeed, in virtue of an isomorphism
$\C_2\simeq\cl_{3,0}$ a general element
\[
\cA=a^0\e_0+\sum^3_{i=1}a^i\e_i+\sum^3_{i=1}\sum^3_{j=1}a^{ij}\e_{ij}+
a^{123}\e_{123}
\]
of the algebra $\cl_{3,0}$ can be written in the form
\begin{equation}\label{6.17'}
\cA=(a^0+\omega a^{123})\e_0+(a^1+\omega a^{23})\e_1+(a^2+\omega
a^{31})\e_2 +(a^3+\omega a^{12})\e_3,
\end{equation}
where $\omega=\e_{123}$. Since $\omega$ belongs to a center of the
algebra $\cl_{3,0}$ ($\omega$ commutes with all the basis elements)
and $\omega^2=-1$, then we can to suppose $\omega\equiv i$. The
action of the automorphism $\star$ on the homogeneous element $\cA$
of the degree $k$ is defined by the formula (\ref{auto16}). In
accordance with this the action of the automorphism
$\cA\rightarrow\cA^\star$, where $\cA$ is the element (\ref{6.17'}),
has the form
\begin{equation}\label{In1}
\cA\longrightarrow\cA^\star=-(a^0-\omega a^{123})\e_0-(a^1-\omega
a^{23})\e_1 -(a^2-\omega a^{31})\e_2-(a^3-\omega a^{12})\e_3.
\end{equation}
Therefore, $\star:\,\C_2\rightarrow -\overset{\ast}{\C}_2$.
Correspondingly, the action of the antiautomorphism
$\cA\rightarrow\widetilde{\cA}$ on the homogeneous element $\cA$ of
the degree $k$ is defined by the formula (\ref{auto19}). Thus, for
the element (\ref{6.17'}) we obtain
\begin{equation}\label{In2}
\cA\longrightarrow\widetilde{\cA}=(a^0-\omega a^{123})\e_0+
(a^1-\omega a^{23})\e_1+(a^2-\omega a^{31})\e_2+(a^3-\omega
a^{12})\e_3,
\end{equation}
that is, $\widetilde{\phantom{cc}}:\,\C_2\rightarrow
\overset{\ast}{\C}_2$. This allows us to define an algebraic
analogue of the Wigner's representation doubling:
$\C_2\oplus\overset{\ast}{\C}_2$. Further, from (\ref{6.17'}) it
follows that
$\cA=\cA_1+\omega\cA_2=(a^0\e_0+a^1\e_1+a^2\e_2+a^3\e_3)+
\omega(a^{123}\e_0+a^{23}\e_1+a^{31}\e_2+a^{12}\e_3)$. In general
case, by virtue of an isomorphism $\C_{2k}\simeq\cl_{p,q}$, where
$\cl_{p,q}$ is a real Clifford algebra with a division ring
$\K\simeq\C$, $p-q\equiv 3,7 \pmod{8}$, we have for the general
element of $\cl_{p,q}$ an expression $\cA=\cA_1+\omega\cA_2$, here
$\omega^2=\e^2_{12\ldots p+q}=-1$ and, therefore, $\omega\equiv i$.
Thus, from $\C_{2k}$ under action of the automorphism
$\cA\rightarrow\cA^\star$ we obtain a general algebraic doubling
\begin{equation}\label{D}
\C_{2k}\oplus\overset{\ast}{\C}_{2k}.
\end{equation}

The tensor product
$\overset{\ast}{\C}_2\otimes\overset{\ast}{\C}_2\otimes\cdots\otimes
\overset{\ast}{\C}_2\simeq\overset{\ast}{\C}_{2r}$ of the $r$
algebras $\overset{\ast}{\C}_2$ induces the tensor product of the
$r$ spinspaces $\dot{\dS}_2$:
\[
\dot{\dS}_2\otimes\dot{\dS}_2\otimes\cdots\otimes\dot{\dS}_2=\dot{\dS}_{2^r}.
\]
Vectors of the spinspace $\dot{\dS}_{2^r}$ have the form
\begin{equation}\label{6.18}
\boldsymbol{s}^{\dot{\alpha}_1\dot{\alpha}_2\cdots\dot{\alpha}_r}=\sum
\boldsymbol{s}^{\dot{\alpha}_1}\otimes
\boldsymbol{s}^{\dot{\alpha}_2}\otimes\cdots\otimes
\boldsymbol{s}^{\dot{\alpha}_r},
\end{equation}
where the each cospinor $\boldsymbol{s}^{\dot{\alpha}_i}$ from
(\ref{6.18}) is transformed by the rule ${}^\prime
\boldsymbol{s}^{\dot{\alpha}^\prime_i}=
\sigma^{\dot{\alpha}^\prime_i}_{\dot{\alpha}_i}\boldsymbol{s}^{\dot{\alpha}_i}$.
Therefore,
\begin{equation}\label{6.19}
{}^\prime
\boldsymbol{s}^{\dot{\alpha}^\prime_1\dot{\alpha}^\prime_2\cdots
\dot{\alpha}^\prime_r}=\sum\sigma^{\dot{\alpha}^\prime_1}_{\dot{\alpha}_1}
\sigma^{\dot{\alpha}^\prime_2}_{\dot{\alpha}_2}\cdots
\sigma^{\dot{\alpha}^\prime_r}_{\dot{\alpha}_r}
\boldsymbol{s}^{\dot{\alpha}_1\dot{\alpha}_2\cdots\dot{\alpha}_r}.
\end{equation}\begin{sloppypar}\noindent
The representation (\ref{6.19}) is called {\it a dotted spintensor
representation of the proper Lorentz group of the rank
$r$}.\end{sloppypar}

In general case we have a tensor product of the $k$ algebras $\C_2$
and the $r$ algebras $\overset{\ast}{\C}_2$:
\[
\C_2\otimes\C_2\otimes\cdots\otimes\C_2\bigotimes
\overset{\ast}{\C}_2\otimes
\overset{\ast}{\C}_2\otimes\cdots\otimes\overset{\ast}{\C}_2\simeq
\C_{2k}\otimes\overset{\ast}{\C}_{2r},
\]
which induces a spinspace
\[
\dS_2\otimes\dS_2\otimes\cdots\otimes\dS_2\bigotimes\dot{\dS}_2\otimes
\dot{\dS}_2\otimes\cdots\otimes\dot{\dS}_2=\dS_{2^{k+r}}
\]
with the vectors
\begin{equation}\label{6.20'}
\boldsymbol{s}^{\alpha_1\alpha_2\cdots\alpha_k\dot{\alpha}_1\dot{\alpha}_2\cdots
\dot{\alpha}_r}=\sum \boldsymbol{s}^{\alpha_1}\otimes
\boldsymbol{s}^{\alpha_2}\otimes\cdots\otimes
\boldsymbol{s}^{\alpha_k}\otimes
\boldsymbol{s}^{\dot{\alpha}_1}\otimes
\boldsymbol{s}^{\dot{\alpha}_2}\otimes\cdots\otimes
\boldsymbol{s}^{\dot{\alpha}_r}.
\end{equation}
In this case we have a natural unification of the representations
(\ref{6.17}) and (\ref{6.19}):
\begin{equation}\label{6.20}
{}^\prime
\boldsymbol{s}^{\alpha^\prime_1\alpha^\prime_2\cdots\alpha^\prime_k
\dot{\alpha}^\prime_1\dot{\alpha}^\prime_2\cdots\dot{\alpha}^\prime_r}=\sum
\sigma^{\alpha^\prime_1}_{\alpha_1}\sigma^{\alpha^\prime_2}_{\alpha_2}\cdots
\sigma^{\alpha^\prime_k}_{\alpha_k}\sigma^{\dot{\alpha}^\prime_1}_{
\dot{\alpha}_1}\sigma^{\dot{\alpha}^\prime_2}_{\dot{\alpha}_2}\cdots
\sigma^{\dot{\alpha}^\prime_r}_{\dot{\alpha}_r}
\boldsymbol{s}^{\alpha_1\alpha_2\cdots\alpha_k\dot{\alpha}_1\dot{\alpha}_2\cdots
\dot{\alpha}_r}.
\end{equation}
So, a representation (\ref{6.20}) is called {\it a spintensor
representation of the proper Lorentz group of the rank $(k,r)$}.

Further, let $\fg\rightarrow T_{\fg}$ be an arbitrary linear
representation of the proper orthochronous Lorentz group
$\fG_+=\SO_0(1,3)$ and let $\sA_i(t)=T_{a_i(t)}$ be an infinitesimal
operator corresponding to the rotation $a_i(t)\in\fG_+$.
Analogously, let $\sB_i(t)=T_{b_i(t)}$, where $b_i(t)\in\fG_+$ is
the hyperbolic rotation. The operators $\sA_i$ and $\sB_i$ satisfy
to the following relations:
\begin{equation}\label{Com1}
\left.\begin{array}{lll} \ld\sA_1,\sA_2\rd=\sA_3, &
\ld\sA_2,\sA_3\rd=\sA_1, &
\ld\sA_3,\sA_1\rd=\sA_2,\\[0.1cm]
\ld\sB_1,\sB_2\rd=-\sA_3, & \ld\sB_2,\sB_3\rd=-\sA_1, &
\ld\sB_3,\sB_1\rd=-\sA_2,\\[0.1cm]
\ld\sA_1,\sB_1\rd=0, & \ld\sA_2,\sB_2\rd=0, &
\ld\sA_3,\sB_3\rd=0,\\[0.1cm]
\ld\sA_1,\sB_2\rd=\sB_3, & \ld\sA_1,\sB_3\rd=-\sB_2, & \\[0.1cm]
\ld\sA_2,\sB_3\rd=\sB_1, & \ld\sA_2,\sB_1\rd=-\sB_3, & \\[0.1cm]
\ld\sA_3,\sB_1\rd=\sB_2, & \ld\sA_3,\sB_2\rd=-\sB_1. &
\end{array}\right\}
\end{equation}

As is known \cite{GMS}, finite-dimensional (spinor) representations
of the group $\SO_0(1,3)$ in the space of symmetrical polynomials
$\Sym_{(k,r)}$ have the following form:
\begin{equation}\label{TenRep}
T_{\fg}q(\xi,\overline{\xi})=(\gamma\xi+\delta)^{l_0+l_1-1}
\overline{(\gamma\xi+\delta)}^{l_0-l_1+1}q\left(\frac{\alpha\xi+\beta}{\gamma\xi+\delta};
\frac{\overline{\alpha\xi+\beta}}{\overline{\gamma\xi+\delta}}\right),
\end{equation}
where $k=l_0+l_1-1$, $r=l_0-l_1+1$, and the pair $(l_0,l_1)$ defines
some representation of the group $\SO_0(1,3)$ in the
Gel'fand-Naimark basis:
\[
H_{3}\xi_{k\nu} =m\xi_{k\nu},
\]
\[
H_{+}\xi_{k\nu} =\sqrt{(k+\nu+1)(k-\nu)}\xi_{k,\nu+1},
\]
\[
H_{-}\xi_{k\nu} =\sqrt{(k+\nu)(k-\nu+1)}\xi_{k,\nu-1},
\]
\[
F_{3}\xi_{k\nu}
=C_{l}\sqrt{k^{2}-\nu^{2}}\xi_{k-1,\nu}-A_{l}\nu\xi_{k,\nu}-C_{k+1}\sqrt{(k+1)^{2}-\nu^{2}}\xi_{k+1,\nu},
\]
\begin{multline}
F_{+}\xi_{k\nu} =C_{k}\sqrt{(k-\nu)(k-\nu-1)}\xi_{k-1,\nu+1}
-A_{k}\sqrt{(k-\nu)(k+\nu+1)}\xi_{k,\nu+1}+ \\
+C_{k+1}\sqrt{(k+\nu+1)(k+\nu+2)}\xi_{k+1,\nu+1},\nonumber
\end{multline}
\begin{multline}
F_{-}\xi_{k\nu} =-C_{k}\sqrt{(k+\nu)(k+\nu-1)}\xi_{k-1,\nu-1}-A_{k}\sqrt{(k+\nu)(k-\nu+1)}\xi_{k,\nu-1}-\\
-C_{k+1}\sqrt{(k-\nu+1)(k-\nu+2)}\xi_{k+1,\nu-1},\nonumber
\end{multline}
\begin{equation}\label{GNB}
A_{k}=\frac{il_{0}l_{1}}{k(k+1)},\quad
C_{k}=\frac{i}{k}\sqrt{\frac{(k^{2}-l^{2}_{0})(k^{2}-l^{2}_{1})}
{4k^{2}-1}},
\end{equation}
$$\nu=-k,-k+1,\ldots,k-1,k,$$
$$k=l_{0}\,,l_{0}+1,\ldots,$$
where $l_{0}$ is positive integer or half-integer number, $l_{1}$ is
an arbitrary complex number. These formulae define a
finite-dimensional representation of the group $\SO_0(1,3)$ when
$l^2_1=(l_0+p)^2$, $p$ is some natural number. In the case
$l^2_1\neq(l_0+p)^2$ we have an infinite-dimensional representation
of $\SO_0(1,3)$. The operators $H_{3},H_{+},H_{-},F_{3},F_{+},F_{-}$
are
\begin{eqnarray}
&&H_+=i\sA_1-\sA_2,\quad H_-=i\sA_1+\sA_2,\quad H_3=i\sA_3,\nonumber\\
&&F_+=i\sB_1-\sB_2,\quad F_-=i\sB_1+\sB_2,\quad F_3=i\sB_3.\nonumber
\end{eqnarray}
Let us consider the operators
\begin{gather}
\sX_l=\frac{1}{2}i(\sA_l+i\sB_l),\quad\sY_l=\frac{1}{2}i(\sA_l-\bi\sB_l),
\label{SL25}\\
(l=1,2,3).\nonumber
\end{gather}
Using the relations (\ref{Com1}), we find that
\begin{equation}\label{Com2}
\ld\sX_k,\sX_l\rd=i\varepsilon_{klm}\sX_m,\quad
\ld\sY_l,\sY_m\rd=i\varepsilon_{lmn}\sY_n,\quad \ld\sX_l,\sY_m\rd=0.
\end{equation}
Further, introducing generators of the form
\begin{equation}\label{SL26}
\left.\begin{array}{cc}
\sX_+=\sX_1+i\sX_2, & \sX_-=\sX_1-i\sX_2,\\[0.1cm]
\sY_+=\sY_1+i\sY_2, & \sY_-=\sY_1-i\sY_2,
\end{array}\right\}
\end{equation}
we see that in virtue of commutativity of the relations (\ref{Com2})
a space of an irreducible finite--dimensional representation of the
group $\SL(2,\C)$ can be spanned on the totality of
$(2l+1)(2\dot{l}+1)$ basis vectors $\mid
l,m;\dot{l},\dot{m}\rangle$, where $l,m,\dot{l},\dot{m}$ are integer
or half--integer numbers, $-l\leq m\leq l$, $-\dot{l}\leq
\dot{m}\leq \dot{l}$. Therefore,
\begin{eqnarray}
&&\sX_-\mid l,m;\dot{l},\dot{m}\rangle= \sqrt{(l+m)(l-m+1)}\mid
l,m-1,\dot{l},\dot{m}\rangle
\;\;(m>-l),\nonumber\\
&&\sX_+\mid l,m;\dot{l},\dot{m}\rangle= \sqrt{(l-m)(l+m+1)}\mid
l,m+1;\dot{l},\dot{m}\rangle
\;\;(m<l),\nonumber\\
&&\sX_3\mid l,m;\dot{l},\dot{m}\rangle=
m\mid l,m;\dot{l},\dot{m}\rangle,\nonumber\\
&&\sY_-\mid l,m;\dot{l},\dot{m}\rangle=
\sqrt{(\dot{l}+\dot{m})(\dot{l}-\dot{m}+1)}\mid
l,m;\dot{l},\dot{m}-1
\rangle\;\;(\dot{m}>-\dot{l}),\nonumber\\
&&\sY_+\mid l,m;\dot{l},\dot{m}\rangle=
\sqrt{(\dot{l}-\dot{m})(\dot{l}+\dot{m}+1)}\mid
l,m;\dot{l},\dot{m}+1
\rangle\;\;(\dot{m}<\dot{l}),\nonumber\\
&&\sY_3\mid l,m;\dot{l},\dot{m}\rangle= \dot{m}\mid
l,m;\dot{l},\dot{m}\rangle.\label{Waerden}
\end{eqnarray}
From the relations (\ref{Com2}) it follows that each of the sets of
infinitesimal operators $\sX$ and $\sY$ generates the group $\SU(2)$
and these two groups commute with each other. Thus, from the
relations (\ref{Com2}) and (\ref{Waerden}) it follows that the group
$\SL(2,\C)$, in essence, is equivalent locally to the group
$\SU(2)\otimes\SU(2)$. This representation for the Lorentz group was
first given by Van der Waerden in \cite{Wa32}. The following
relations between generators $\sY_\pm$, $\sX_\pm$, $\sY_3$, $\sX_3$
and $H_\pm$, $F_\pm$, $H_3$, $F_3$ define a relationship between the
Van der Waerden and Gel'fand-Naimark bases:
\[
{\renewcommand{\arraystretch}{1.7}
\begin{array}{ccc}
\sY_+&=&-\dfrac{1}{2}(F_++i H_+),\\
\sY_-&=&-\dfrac{1}{2}(F_-+i H_-),\\
\sY_3&=&-\dfrac{1}{2}(F_3+i H_3),
\end{array}\quad
\begin{array}{ccc}
\sX_+&=&\dfrac{1}{2}(F_+-i H_+),\\
\sX_-&=&\dfrac{1}{2}(F_--i H_-),\\
\sX_3&=&\dfrac{1}{2}(F_3-i H_3).
\end{array}
}
\]
The relation between the numbers $l_0$, $l_1$ and the number $l$
(the weight of representation in the basis (\ref{Waerden})) is given
by the following formula:
\[
(l_0,l_1)=\left(l,l+1\right).
\]
Whence it immediately follows that
\begin{equation}\label{RelL}
l=\frac{l_0+l_1-1}{2}.
\end{equation}
As is known \cite{GMS}, if an irreducible representation of the
proper Lorentz group $\SO_0(1,3)$ is defined by the pair
$(l_0,l_1)$, then a conjugated representation is also irreducible
and is defined by a pair $\pm(l_0,-l_1)$. Therefore,
\[
(l_0,l_1)=\left(-\dot{l},\,\dot{l}+1\right).
\]
Thus,
\begin{equation}\label{RelDL}
\dot{l}=\frac{l_0-l_1+1}{2}.
\end{equation}

Further, representations $\fC^{s_1,s_2}$ and
$\fC^{s^\prime_1,s^\prime_2}$ are called \emph{interlocking
irreducible representations of the Lorentz group }, that is, such
representations that $s^\prime_1=s_1\pm 1$, $s^\prime_2=s_2\pm 1$
\cite{GY48}. The two most full schemes of the interlocking
irreducible representations of the Lorentz group (Gel'fand-Yaglom
chains) for integer and half-integer spins are shown on the Fig.\,3
and Fig.\,4.
\begin{figure}[ht]
\[
\dgARROWLENGTH=0.5em \dgHORIZPAD=1.7em \dgVERTPAD=2.2ex
\begin{diagram}
\node[5]{\fC^{s,0}}\arrow{e,-}\arrow{s,-}\node{\cdots}\\
\node[5]{\vdots}\arrow{s,-}\\
\node[3]{\fC^{4,0}}\arrow{e,-}\arrow{s,-}\node{\cdots}\arrow{e,-}
\node{\fC^{s+2,2-s}}\arrow{s,-}\arrow{e,-}
\node{\cdots}\\
\node[2]{\fC^{2,0}}\arrow{s,-}\arrow{e,-}
\node{\fC^{3,-1}}\arrow{s,-}\arrow{e,-} \node{\cdots}\arrow{e,-}
\node{\fC^{s+1,1-s}}\arrow{s,-}\arrow{e,-}
\node{\cdots}\\
\node{\fC^{0,0}}\arrow{e,-} \node{\fC^{1,-1}}\arrow{s,-}\arrow{e,-}
\node{\fC^{2,-2}}\arrow{s,-}\arrow{e,-}\node{\cdots}\arrow{e,-}
\node{\fC^{s,-s}}\arrow{s,-}\arrow{e,-}
\node{\cdots}\\
\node[2]{\fC^{0,-2}}\arrow{e,-}
\node{\fC^{1,-3}}\arrow{s,-}\arrow{e,-} \node{\cdots}\arrow{e,-}
\node{\fC^{s-1,-1-s}}\arrow{s,-}\arrow{e,-}
\node{\cdots}\\
\node[3]{\fC^{0,-4}}\arrow{e,-}\node{\cdots}\arrow{e,-}
\node{\fC^{s-2,-2-s}}\arrow{s,-}\arrow{e,-}
\node{\cdots}\\
\node[5]{\vdots}\arrow{s,-}\\
\node[5]{\fC^{0,-2s}}\arrow{e,-}\node{\cdots}
\end{diagram}
\]
\begin{center}{\small {\bf Fig.\,3:} Complex interlocking representations of $\spin_+(1,3)$ for the fields of integer spin
(Bose-scheme).}\end{center}
\end{figure}
\begin{figure}[ht]
\[
\dgARROWLENGTH=0.5em
\dgHORIZPAD=1.7em 
\dgVERTPAD=2.2ex 
\begin{diagram}
\node[4]{\fC^{2s,0}}\arrow{s,-}\arrow{e,-}\node{\cdots}\\
\node[4]{\vdots}\arrow{s,-}\\
\node[2]{\fC^{3,0}}\arrow{s,-}\arrow{e,-} \node{\cdots}\arrow{e,-}
\node{\fC^{(2s+3)/2,(3-2s)2}}\arrow{s,-}\arrow{e,-}
\node{\cdots}\\
\node{\fC^{1,0}}\arrow{s,-}\arrow{e,-}
\node{\fC^{2,-1}}\arrow{s,-}\arrow{e,-} \node{\cdots}\arrow{e,-}
\node{\fC^{(2s+1)/2,(1-2s)/2}}\arrow{s,-}\arrow{e,-}
\node{\cdots}\\
\node{\fC^{0,-1}}\arrow{e,-} \node{\fC^{1,-2}}\arrow{s,-}\arrow{e,-}
\node{\cdots}\arrow{e,-}
\node{\fC^{(2s-1)/2,-(2s+1)/2}}\arrow{s,-}\arrow{e,-}
\node{\cdots}\\
\node[2]{\fC^{0,-3}}\arrow{e,-} \node{\cdots}\arrow{e,-}
\node{\fC^{(2s-3)/2,-(2s+3)/2}}\arrow{s,-}\arrow{e,-}
\node{\cdots}\\
\node[4]{\vdots}\arrow{s,-}\\
\node[4]{\fC^{0,-2s}}\arrow{e,-}\node{\cdots}
\end{diagram}
\]
\begin{center}{\small {\bf Fig.\,4:} Complex interlocking representations of $\spin_+(1,3)$ for the fields of half-integer spin
(Fermi-scheme).}\end{center}
\end{figure}
As follows from Fig.\,3 the simplest field is the scalar field
$\fC^{0,0}$. This field is described by the Fock-Klein-Gordon
equation. In its turn, the simplest field from the Fermi-scheme
(Fig.\,4) is the electron-positron (spinor) field corresponding to
the following interlocking scheme:
\[
\dgARROWLENGTH=2.5em
\begin{diagram}
\node{\fC^{1,0}}\arrow{e,<>} \node{\fC^{0,-1}}
\end{diagram}.
\]
\begin{sloppypar}\noindent
This field is described by the Dirac equation. Further, the next
field from the Bose-scheme (Fig.\,3) is a photon field (Maxwell
field) defined within the interlocking scheme\end{sloppypar}
\[
\dgARROWLENGTH=2.5em
\begin{diagram}
\node{\fC^{2,0}}\arrow{e,<>}\node{\fC^{1,-1}}
\arrow{e,<>}\node{\fC^{0,-2}}
\end{diagram}.
\]
This interlocking scheme leads to the Maxwell equations. The fields
$(1/2,0)\oplus(0,1/2)$ and $(1,0)\oplus(0,1)$ (Dirac and Maxwell
fields) are particular cases of fields of the type
$(l,0)\oplus(0,l)$. Wave equations for such fields and their general
solutions were found in the works \cite{Var03,Var04e,Var05b}.

It is easy to see that the interlocking scheme, corresponded to the
Maxwell field, contains the field of tensor type, $\fC^{1,-1}$.
Further, the next interlocking scheme (see Fig.\,4)
\[
\dgARROWLENGTH=2.5em
\begin{diagram}
\node{\fC^{3,0}}\arrow{e,<>}\node{\fC^{2,-1}}
\arrow{e,<>}\node{\fC^{1,-2}}\arrow{e,<>}\node{\fC^{0,-3}}
\end{diagram},
\]
corresponding to the Pauli-Fierz equations \cite{FP39}, contains a
chain of the type
\[
\dgARROWLENGTH=2.5em
\begin{diagram}
\node{\fC^{2,-1}}\arrow{e,<>} \node{\fC^{1,-2}}
\end{diagram}.
\]
In such a way we come to wave equations for the fields
$\boldsymbol{\psi}(\balpha)=\langle
x,\fg\,|\boldsymbol{\psi}\rangle$ of tensor type
$(l_1,l_2)\oplus(l_2,l_1)$. Wave equations for such fields and their
general solutions were found in the work \cite{Var07b}.

A relation between the numbers $l_0$, $l_1$ of the Gel'fand-Naimark
representation (\ref{GNB}) and the number $k$ of the factors $\C_2$
in the product $\C_2\otimes\C_2\otimes\cdots\otimes\C_2$ is given by
the following formula:
\[
(l_0,l_1)=\left(\frac{k}{2},\frac{k}{2}+1\right),
\]
Hence it immediately follows that $k=l_0+l_1-1$. Thus, we have {\it
a complex representation $\fC^{l_0+l_1-1,0}$  of the group
$\spin_+(1,3)$ in the spinspace $\dS_{2^k}$}. If the representation
$\fC^{l_0+l_1-1,0}$ is reducible, then the space $\dS_{2^{k}}$ is
decomposed into a direct sum of irreducible subspaces, that is, it
is possible to choose in $\dS_{2^{k}}$ such a basis, in which all
the matrices take a block-diagonal form. Then the field
$\boldsymbol{\psi}(\balpha)$ is reduced to some number of the fields
corresponding to irreducible representations of the group
$\spin_+(1,3)$, each of which is transformed independently from the
other, and the field $\boldsymbol{\psi}(\balpha)$ in this case is a
collection of the fields with more simple structure. It is obvious
that these more simple fields correspond to irreducible
representations $\fC$.

Analogously, a relation between the numbers $l_0$, $l_1$ of the
Gel'fand-Naimark representation (\ref{GNB}) and the number $r$ of
the factors $\overset{\ast}{\C}_2$ in the product
$\overset{\ast}{\C}_2\otimes
\overset{\ast}{\C}_2\otimes\cdots\otimes\overset{\ast}{\C}_2$ is
given by the following formula:
\[
(-l_0,l_1)=\left(-\frac{r}{2},\frac{r}{2}+1\right).
\]
Hence it immediately follows that $r=l_0-l_1+1$. Thus, we have a
complex representation $\fC^{0,l_0-l_1+1}$ of $\spin_+(1,3)$ in the
spinspace $\dS_{2^r}$.

As is known \cite{Nai58,GMS,RF}, a system of irreducible
finite-dimensional representations of the group $\fG_+$ is realized
in the space $\Sym_{(k,r)}\subset\dS_{2^{k+r}}$ of symmetric
spintensors. The dimensionality of $\Sym_{(k,r)}$ is equal to
$(k+1)(r+1)$. A representation of the group $\fG_+$, defined by such
spintensors, is irreducible and denoted by the symbol
$\fD^{(l,\dot{l})}(\sigma)$, where $2l=k,\;2\dot{l}=r$, the numbers
$l$ and $\dot{l}$ are integer or half-integer. In general case, the
field $\boldsymbol{\psi}(\balpha)$ is the field of type
$(l,\dot{l})$. As a rule, in physics there are three basic types of
the fields.\\[0.3cm]
a) The field of type $(l,0)$. The structure of this field (or the
field $(0,\dot{l})$) is described by the representation
$\fD^{(l,0)}(\sigma)$ ($\fD^{(0,\dot{l})}(\sigma)$), which is
realized in the space $\dS_{2^k}$ ($\dS_{2^r}$). At this point, the
algebra $\C_{2k}\simeq\C_2\otimes\C_2\otimes\cdots\otimes\C_2$
(correspondingly,
$\overset{\ast}{\C}_{2k}\simeq\overset{\ast}{\C}_2\otimes
\overset{\ast}{\C}_2\otimes\cdots\otimes\overset{\ast}{\C}_2$) is
associated with the field of the type $(l,0)$ (correspondingly,
$(0,\dot{l})$). The trivial case $l=0$ corresponds to {\it a
Pauli-Weisskopf field} describing the scalar particles. Further, at
$l=\dot{l}=1/2$ we have {\it a Weyl field} describing the neutrino.
At this point, the antineutrino is described by a fundamental
representation $\fD^{(1/2,0)}(\sigma)=\sigma$ of the group $\fG_+$
and the algebra $\C_2$. Correspondingly, the neutrino is described
by a conjugated representation $\fD^{(0,1/2)}(\sigma)$ and the
algebra $\overset{\ast}{\C}_2$. In essence, one can say that the
algebra $\C_2$ ($\overset{\ast}{\C}_2$) is the basic building block,
from which other physical fields built by means of direct sum or
tensor product. One can say also that this situation looks like
the de Broglie fusion method \cite{Bro43}.\\[0.3cm]
b) The field of type $(l,0)\oplus(0,\dot{l})$. The structure of this
field admits a space inversion and, therefore, in accordance with a
Wigner's doubling \cite{Wig64} is described by a representation
$\fD^{(l,0)}\oplus \fD^{(0,\dot{l})}$ of the group $\fG_+$. This
representation is realized in the space $\dS_{2^{2k}}$. The Clifford
algebra, related with this representation, is a direct sum
$\C_{2k}\oplus\overset{\ast}{\C}_{2k}\simeq
\C_2\otimes\C_2\otimes\cdots\otimes\C_2\bigoplus\overset{\ast}{\C}_2\otimes
\overset{\ast}{\C}_2\otimes\cdots\otimes\overset{\ast}{\C}_2$. In
the simplest case $l=1/2$ we have {\it bispinor (electron--positron)
Dirac field} $(1/2,0)\oplus(0,1/2)$ with the algebra $\C_2\oplus
\overset{\ast}{\C}_2$. It should be noted that the Dirac algebra
$\C_4$, considered as a tensor product $\C_2\otimes\C_2$ (or
$\C_2\otimes\overset{\ast}{\C}_2$) in accordance with (\ref{6.16})
(or (\ref{6.20'})) gives rise to spintensors
$\boldsymbol{s}^{\alpha_1\alpha_2}$ (or
$\boldsymbol{s}^{\alpha_1\dot{\alpha}_1}$), but it contradicts with
the usual definition of the Dirac bispinor as a pair
$(\boldsymbol{s}^{\alpha_1},\boldsymbol{s}^{\dot{\alpha}_1})$.
Therefore, the Clifford algebra, associated with the Dirac field, is
$\C_2\oplus\overset{\ast}{\C}_2$, and a spinspace of this sum in
virtue of unique decomposition $\dS_2\oplus\dot{\dS}_2=\dS_4$ is a
spinspace of $\C_4$.\\[0.3cm]
c) Tensor fields
$(l^\prime,l^{\prime\prime})\oplus(\dot{l}^{\prime\prime},\dot{l}^\prime)$.
The fields $(l^\prime,l^{\prime\prime})$ and
$(\dot{l}^{\prime\prime},\dot{l}^\prime)$ are defined within the
arbitrary spin chains (see Fig.\,3 and Fig.\,4). Universal coverings
of these spin chains are constructed within the representations
$\fC^{l_0+l_1-1,l_0-l_1+1}$ and $\fC^{l_0-l_1+1,l_0+l_1-1}$ of
$\spin_+(1,3)$ in the spinspaces $\dS_{2^{k+r}}$ associated with the
algebra $\C_2\otimes\C_2\otimes\cdots\otimes\C_2\bigotimes
\overset{\ast}{\C}_2\otimes
\overset{\ast}{\C}_2\otimes\cdots\otimes\overset{\ast}{\C}_2$. A
relation between the numbers $l_0$, $l_1$ of the Gel'fand-Naimark
basis (\ref{GNB}) and the numbers $k$ and $r$ of the factors $\C_2$
and $\overset{\ast}{\C}_2$ is given by the following formula:
\[
(l_0,l_1)=\left(\frac{k-r}{2},\frac{k+r}{2}+1\right).
\]
2) Real representations.\\
As is known \cite{GMS}, if an irreducible representation of the
proper Lorentz group $\SO_0(1,3)$ is defined by the pair
$(l_0,l_1)$, then a conjugated representation is also irreducible
and is defined by a pair $\pm(l_0,-l_1)$. Hence it follows that an
irreducible representation is equivalent to its conjugated
representation only in the case when this representation is defined
by a pair $(0,l_1)$ or $(l_0,0)$, that is, one from the two numbers
$l_0$ and $l_1$ is equal to zero. We take $l_1=0$. In its turn, for
the complex Clifford algebra $\C_n$ ($\overset{\ast}{\C}_n$),
associated with the representation
$\fC^{l_0+l_1-1,0}\,(\fC^{0,l_0-l_1+1})$ of $\fG_+$, the equivalence
of the representation to its conjugated representation induce a
relation $\C_n=\overset{\ast}{\C}_n$, which is fulfilled only in the
case when the algebra $\C_n=\C\otimes\cl_{p,q}$ is reduced to its
real subalgebra $\cl_{p,q}$ ($p+q=n$). Thus, the restriction of the
complex representation $\fC^{l_0+l_1-1,0}$ (or $\fC^{0,l_0-l_1+1}$)
of the group $\spin_+(1,3)$ onto the real representation,
$(l_0,l_1)\rightarrow(l_0,0)$, induce a restriction
$\C_n\rightarrow\cl_{p,q}$. Further, as is known, over the field
$\F=\R$ at $p+q\equiv 0\pmod{2}$ there exist four types of real
subalgebras $\cl_{p,q}$: two types $p-q\equiv 0,2\pmod{8}$ with the
real division ring $\K\simeq\R$ and two types $p-q\equiv
4,6\pmod{8}$ with the quaternionic division ring $\K\simeq\BH$.
Therefore, we have the following four classes of real
representations of the group $\spin_+(1,3)$:
\[
\fR^{l_0}_0\;\leftrightarrow\;\cl_{p,q},\;p-q\equiv
0\pmod{8},\;\K\simeq\R;
\]
\[
\fR^{l_0}_2\;\leftrightarrow\;\cl_{p,q},\;p-q\equiv
2\pmod{8},\;\K\simeq\R;
\]
\[
\fH^{l_0}_4\;\leftrightarrow\;\cl_{p,q},\;p-q\equiv
4\pmod{8},\;\K\simeq\BH;
\]
\begin{equation}
\fH^{l_0}_6\;\leftrightarrow\;\cl_{p,q},\;p-q\equiv
6\pmod{8},\;\K\simeq\BH. \label{IdentPer}
\end{equation}
We call the representations $\fH^{l_0}_4$ and $\fH^{l_0}_6$ {\it
quaternionic representations} of the group $\spin_+(1,3)$. On the
other hand, over the field $\F=\R$ at $p+q\equiv 1\pmod{2}$ there
exist four types of real subalgebras $\cl_{p,q}$: two types
$p-q\equiv 3,7\pmod{8}$ with the complex division ring $\K\simeq\C$,
one type $p-q\equiv 1\pmod{8}$ with the double division ring
$\K\simeq\R\oplus\R$ and one type $p-q\equiv 5\pmod{8}$ with the
double quaternionic division ring $\K\simeq\BH\oplus\BH$. Therefore,
we have the following four classes of real representations of
$\spin_+(1,3)$:
\[
\fC^{l_0}_3\;\leftrightarrow\;\cl_{p,q},\;p-q\equiv
3\pmod{8},\;\K\simeq\C;
\]
\[
\fC^{l_0}_7\;\leftrightarrow\;\cl_{p,q},\;p-q\equiv
7\pmod{8},\;\K\simeq\C;
\]
\[
\fR^{l_0}_{0,2}\cup\fR^{l_0}_{0,2}\;\leftrightarrow\;\cl_{p,q},\;p-q\equiv
1\pmod{8},\;\K\simeq\R\oplus\R;
\]
\begin{equation}\label{IdentPerD}
\fH^{l_0}_{4,6}\cup\fH^{l_0}_{4,6}\;\leftrightarrow\;\cl_{p,q},\;p-q\equiv
5\pmod{8},\;\K\simeq\BH\oplus\BH.
\end{equation}
Here $\fR^{l_0}_{0,2}\cup\fR^{l_0}_{0,2}\simeq\fR^{l_0}_1$ means
that in virtue of an isomorphism
$\cl_{p,q}\simeq\cl_{p,q-1}\oplus\cl_{p,q-1}$ (or
$\cl_{p,q}\simeq\cl_{q,p-1}\oplus\cl_{q,p-1}$), where $p-q\equiv
1\pmod{8}$ and algebras $\cl_{p,q-1}$ (or $\cl_{q,p-1}$) are of the
type $p-q\equiv 0\pmod{8}$ (or $p-q\equiv 2\pmod{8}$), a
representation $\fR^{l_0}_{1}$ is equivalent to
$\fR^{l_0}_{0}\cup\fR^{l_0}_{0}$ (or
$\fR^{l_0}_{2}\cup\fR^{l_0}_{2}$). Analogously, we have
$\fH^{l_0}_{4,6}\cup\fH^{l_0}_{4,6}\simeq\fH^{l_0}_5$ for
representations of $\spin_+(1,3)$ with the double quaternionic
division ring $\K\simeq\BH\oplus\BH$. Therefore, all the real
representations of $\spin_+(1,3)$ are divided into the two sets:
$\fM^+$ ($p-q\equiv 0,2,4,6\pmod{8}$) and $\fM^-$ ($p-q\equiv
1,3,5,7\pmod{8}$) which form a full system $\fM=\fM^+\oplus\fM^-$ of
the real representations of $\spin_+(1,3)$. Let us find a relation
of the number $l_0$ with the dimension of the real algebra
$\cl_{p,q}$. Using Karoubi Theorem, we obtain for the algebra
$\cl_{p,q}$ ($p+q\equiv 0\pmod{2}$) the following factorization:
\begin{equation}\label{Ten}
\cl_{p,q}\simeq\underbrace{\cl_{s_i,t_j}\otimes\cl_{s_i,t_j}\otimes\cdots
\otimes\cl_{s_i,t_j}}_{r\;\text{times}},
\end{equation}
where $s_i,t_j\in\{0,1,2\}$. The algebra (\ref{Ten}) corresponds to
the representations $\fR^{l_0}_0$, $\fR^{l_0}_2$, $\fH^{l_0}_4$,
$\fH^{l_0}_6$ of $\spin_+(1,3)$. It is obvious that $l_0=r/2$ and
$n=2r=p+q=4l_0$, therefore, $l_0=(p+q)/4$. When $p+q\equiv
1\pmod{2}$ we obtain
\begin{equation}\label{TenD}
\cl_{p,q}\simeq\underbrace{\cl_{s_i,t_j}\otimes\cl_{s_i,t_j}\otimes\cdots
\otimes\cl_{s_i,t_j}}_{r\;\text{times}}\bigcup\underbrace{\cl_{s_i,t_j}\otimes\cl_{s_i,t_j}\otimes\cdots
\otimes\cl_{s_i,t_j}}_{r\;\text{times}}.
\end{equation}
In its turn, the algebra (\ref{TenD}) is associated with the double
representations $\fR^{l_0}_{0}\cup\fR^{l_0}_{0}$,
$\fR^{l_0}_{2}\cup\fR^{l_0}_{2}$ and
$\fH^{l_0}_{4}\cup\fH^{l_0}_{4}$, $\fH^{l_0}_{6}\cup\fH^{l_0}_{6}$
of the group $\spin_+(1,3)$.
\end{proof}

\section{Modulo 2 and 8 periodicities in particle physics}
As is known, the Atiyah-Bott-Shapiro periodicity is based on the
$\dZ_2$-graded structure of the Clifford algebras $\cl$.
$\dZ_2$-graded Clifford algebras and Brauer-Wall supergroups,
considered in the section 2, allow one to define some cyclic
relations on the representation system of $\fG_+\simeq\spin_+(1,3)$,
which lead to interesting applications in particle physics (in some
sense like to the Gell-Mann--Ne'eman eightfold way \cite{GN64}).\\
1) Complex representations.\\
Graded central simple Clifford algebras over the field $\F=\C$ form
two similarity classes, which, as it is easy to see, coincide with
the two types of $\C_n$: $n\equiv 0,1\pmod{2}$. The set of these two
types (classes) forms a Brauer-Wall supergroup
$BW_{\C}=G(\C_n,\gamma,\bcirc)$ \cite{Wal64,Lou81}, an action of
which is formally equivalent to the action of the cyclic group
$\dZ_2$ (see section 2). Returning to the representations of $\fG_+$
(Theorem 3), we see that in virtue of correspondences
$\C_n\leftrightarrow\fC$ ($n\equiv 0\pmod{2}$) and
$\C_n\leftrightarrow\fC\cup\fC$ ($n\equiv 1\pmod{2}$) all the
complex representations of the group $\fG_+$ form the full system
$\fM=\fM^0\oplus\fM^1$, where $\fM^0$ ($n\equiv 0\pmod{2}$) and
$\fM^1$ ($n\equiv 1\pmod{2}$). The group action of $BW_{\C}$ can be
transferred onto the system $\fM=\fM^0\oplus\fM^1$. Indeed, a cyclic
structure of the supergroup $BW_{\C}$ is defined by the transition
$\C^+_n\overset{h}{\longrightarrow}\C_n$, where the type of $\C_n$
is defined by the formula $n=h+2r$, here $h\in\{0,1\}$, $r\in\dZ$
\cite{BTr87,BT88}. Therefore, the action of $BW_{\C}$ on the system
$\fM$ is defined by a transition
$\overset{+}{\fC}\overset{h}{\longrightarrow} \fC$, where
$\overset{+}{\fC}\!{}^{l_0+l_1-1,0}\simeq\fC^{l_0+l_1-2,0}$ when
$\fC\in\fM^0$ ($h=1$) and
$\overset{+}{\fC}=\left(\fC^{l_0+l_1-1,0}\cup
\fC^{l_0+l_1-1,0}\right)^+\sim\fC^{l_0+l_1-1,0}$ when $\fC\in\fM^1$
($h=0$), $\dim\fC=h+2r$ ($\dim\fC=l_0+l_1-1$ if $\fC\in\fM^0$ and
$\dim\fC=2(l_0+l_1-1)$ if $\fC\in\fM^1$). For example, in virtue of
the isomorphism $\C^+_2\simeq\C_1$ the transition
$\C^+_2\rightarrow\C_2$ ($\C_1\rightarrow \C_2$) induces on the
system $\fM$ a transition $\overset{+}{\fC}\!{}^{1,0}
\rightarrow\fC^{1,0}$ that in virtue of the isomorphism
$\overset{+}{\fC}\!{}^{1,0}\simeq \fC^{0,0}$ is equivalent to
$\fC^{0,0}\rightarrow\fC^{1,0}$ ($\fC^{0,0}$ is the one-dimensional
representation of the group $\fG_+$) and, therefore, $h=1$. In its
turn, the transition $\C^+_3\rightarrow\C_3$ ($\C_2\rightarrow\C_3$)
induces on the system $\fM$ a transition
$\left(\fC^{1,0}\cup\fC^{1,0}\right)^+\rightarrow
\fC^{1,0}\cup\fC^{1,0}$ and, therefore, $h=0$. Thus, we see that a
cyclic structure of the supergroup $BW_{\C}$ induces modulo 2
periodic relations on the system $\fM$, which can be explicitly
shown on the Trautman diagram (see Fig.\,5).
\begin{figure}
\[
\unitlength=0.5mm
\begin{picture}(50.00,60.00)(0,0)
\put(5,25){0} \put(42,25){1} \put(22,-4){$\fC$}
\put(17,55){$\fC\cup\fC$} \put(3,-13){$n\equiv 0\!\!\!\!\pmod{2}$}
\put(3,64){$n\equiv 1\!\!\!\!\pmod{2}$}
\put(20,49.49){$\cdot$} \put(19.5,49.39){$\cdot$}
\put(19,49.27){$\cdot$} \put(18.5,49.14){$\cdot$}
\put(18,49){$\cdot$} \put(17.5,48.85){$\cdot$}
\put(17,48.68){$\cdot$} \put(16.5,48.51){$\cdot$}
\put(16,48.32){$\cdot$} \put(15.5,48.12){$\cdot$}
\put(15,47.91){$\cdot$} \put(14.5,47.69){$\cdot$}
\put(14,47.45){$\cdot$} \put(13.5,47.2){$\cdot$}
\put(13,46.93){$\cdot$} \put(12.5,46.65){$\cdot$}
\put(12,46.35){$\cdot$} \put(11.5,46.04){$\cdot$}
\put(11,45.71){$\cdot$} \put(10.5,45.36){$\cdot$}
\put(10,45){$\cdot$} \put(9.5,44.61){$\cdot$} \put(9,44.21){$\cdot$}
\put(8.5,43.78){$\cdot$} \put(8,43.33){$\cdot$}
\put(7.5,42.85){$\cdot$} \put(7,42.35){$\cdot$}
\put(6.5,41.81){$\cdot$} \put(6,41.25){$\cdot$}
\put(5.5,40.64){$\cdot$} \put(5,40){$\cdot$} \put(4.5,39.3){$\cdot$}
\put(4,38.56){$\cdot$} \put(3.5,37.76){$\cdot$}
\put(3,36.87){$\cdot$} \put(2.5,35.89){$\cdot$}
\put(2,34.79){$\cdot$} \put(1.5,33.53){$\cdot$} \put(1,32){$\cdot$}
\put(0.5,29.97){$\cdot$}
\put(30,49.49){$\cdot$} \put(30.5,49.39){$\cdot$}
\put(31,49.27){$\cdot$} \put(31.5,49.14){$\cdot$}
\put(32,49){$\cdot$} \put(32.5,48.85){$\cdot$}
\put(33,48.68){$\cdot$} \put(33.5,48.51){$\cdot$}
\put(34,48.32){$\cdot$} \put(34.5,48.12){$\cdot$}
\put(35,47.91){$\cdot$} \put(35.5,47.69){$\cdot$}
\put(36,47.45){$\cdot$} \put(36.5,47.2){$\cdot$}
\put(37,46.93){$\cdot$} \put(37.5,46.65){$\cdot$}
\put(38,46.35){$\cdot$} \put(38.5,46.04){$\cdot$}
\put(39,45.71){$\cdot$} \put(39.5,45.36){$\cdot$}
\put(40,45){$\cdot$} \put(40.5,44.61){$\cdot$}
\put(41,44.21){$\cdot$} \put(41.5,43.78){$\cdot$}
\put(42,43.33){$\cdot$} \put(42.5,42.85){$\cdot$}
\put(43,42.35){$\cdot$} \put(43.5,41.81){$\cdot$}
\put(44,41.25){$\cdot$} \put(44.5,40.64){$\cdot$}
\put(45,40){$\cdot$} \put(45.5,39.3){$\cdot$}
\put(46,38.56){$\cdot$} \put(46.5,37.76){$\cdot$}
\put(47,36.87){$\cdot$} \put(47.5,35.89){$\cdot$}
\put(48,34.79){$\cdot$} \put(48.5,33.53){$\cdot$}
\put(49,32){$\cdot$} \put(49.5,29.97){$\cdot$}
\put(0,25){$\cdot$} \put(0,24.5){$\cdot$} \put(0.02,24){$\cdot$}
\put(0.04,23.5){$\cdot$} \put(0.08,23){$\cdot$}
\put(0.12,22.5){$\cdot$} \put(0.18,22){$\cdot$}
\put(0.25,21.5){$\cdot$} \put(0.32,21){$\cdot$}
\put(0.4,20.5){$\cdot$} \put(0.5,20){$\cdot$}
\put(0.61,19.5){$\cdot$} \put(0.73,19){$\cdot$}
\put(0.85,18.5){$\cdot$} \put(1,18){$\cdot$}
\put(1.15,17.5){$\cdot$} \put(1.31,17){$\cdot$}
\put(1.49,16.5){$\cdot$} \put(1.68,16){$\cdot$}
\put(1.88,15.5){$\cdot$} \put(2.09,15){$\cdot$}
\put(2.31,14.5){$\cdot$} \put(2.55,14){$\cdot$}
\put(2.8,13.5){$\cdot$} \put(3.06,13){$\cdot$} \put(0,25.5){$\cdot$}
\put(0.02,26){$\cdot$} \put(0.04,26.5){$\cdot$}
\put(0.08,27){$\cdot$} \put(0.12,27.5){$\cdot$}
\put(0.18,28){$\cdot$} \put(0.25,28.5){$\cdot$}
\put(0.32,29){$\cdot$} \put(0.4,29.5){$\cdot$} \put(0.5,30){$\cdot$}
\put(0.61,30.5){$\cdot$} \put(0.73,31){$\cdot$}
\put(0.85,31.5){$\cdot$} \put(1,32){$\cdot$}
\put(1.15,32.5){$\cdot$} \put(1.31,33){$\cdot$}
\put(1.49,33.5){$\cdot$} \put(1.68,34){$\cdot$}
\put(1.88,34.5){$\cdot$} \put(2.09,35){$\cdot$}
\put(2.31,35.5){$\cdot$} \put(2.55,36){$\cdot$}
\put(2.8,36.5){$\cdot$} \put(3.06,37){$\cdot$}
\put(50,25){$\cdot$} \put(49.99,24.5){$\cdot$}
\put(49.98,24){$\cdot$} \put(49.95,23.5){$\cdot$}
\put(49.92,23){$\cdot$} \put(49.87,22.5){$\cdot$}
\put(49.82,22){$\cdot$} \put(49.75,21.5){$\cdot$}
\put(49.68,21){$\cdot$} \put(49.51,20.5){$\cdot$}
\put(49.49,20){$\cdot$} \put(49.39,19.5){$\cdot$}
\put(49.27,19){$\cdot$} \put(49.14,18.5){$\cdot$}
\put(49,18){$\cdot$} \put(48.85,17.5){$\cdot$}
\put(48.69,17){$\cdot$} \put(48.51,16.5){$\cdot$}
\put(48.32,16){$\cdot$} \put(48.12,15.5){$\cdot$}
\put(47.91,15){$\cdot$} \put(47.69,14.5){$\cdot$}
\put(47.45,14){$\cdot$} \put(47.2,13.5){$\cdot$}
\put(46.93,13){$\cdot$} \put(50,25){$\cdot$}
\put(49.99,25.5){$\cdot$} \put(49.98,26){$\cdot$}
\put(49.95,26.5){$\cdot$} \put(49.92,27){$\cdot$}
\put(49.87,27.5){$\cdot$} \put(49.82,28){$\cdot$}
\put(49.75,28.5){$\cdot$} \put(49.68,29){$\cdot$}
\put(49.51,29.5){$\cdot$} \put(49.49,30){$\cdot$}
\put(49.39,30.5){$\cdot$} \put(49.27,31){$\cdot$}
\put(49.14,31.5){$\cdot$} \put(49,32){$\cdot$}
\put(48.85,32.5){$\cdot$} \put(48.69,33){$\cdot$}
\put(48.51,33.5){$\cdot$} \put(48.32,34){$\cdot$}
\put(48.12,34.5){$\cdot$} \put(47.91,35){$\cdot$}
\put(47.69,35.5){$\cdot$} \put(47.45,36){$\cdot$}
\put(47.2,36.5){$\cdot$} \put(46.93,37){$\cdot$}
\put(20,0.5){$\cdot$} \put(19.5,0.61){$\cdot$}
\put(19,0.73){$\cdot$} \put(18.5,0.86){$\cdot$} \put(18,1){$\cdot$}
\put(17.5,1.15){$\cdot$} \put(17,1.31){$\cdot$}
\put(16.5,1.49){$\cdot$} \put(16,1.68){$\cdot$}
\put(15.5,1.87){$\cdot$} \put(15,2.09){$\cdot$}
\put(14.5,2.31){$\cdot$} \put(14,2.55){$\cdot$}
\put(13.5,2.8){$\cdot$} \put(13,3.06){$\cdot$}
\put(12.5,3.35){$\cdot$} \put(12,3.64){$\cdot$}
\put(11.5,3.96){$\cdot$} \put(11,4.29){$\cdot$}
\put(10.5,4.63){$\cdot$} \put(10,5){$\cdot$} \put(9.5,5.38){$\cdot$}
\put(9,5.79){$\cdot$} \put(8.5,6.22){$\cdot$} \put(8,6.67){$\cdot$}
\put(7.5,7.15){$\cdot$} \put(7,7.65){$\cdot$}
\put(6.5,8.18){$\cdot$} \put(6,8.75){$\cdot$}
\put(5.5,9.35){$\cdot$} \put(5,10){$\cdot$} \put(4.5,10.69){$\cdot$}
\put(4,11.43){$\cdot$} \put(3.5,12.24){$\cdot$}
\put(3,13.12){$\cdot$} \put(2.5,14.10){$\cdot$}
\put(2,15.20){$\cdot$} \put(1.5,16.47){$\cdot$} \put(1,18){$\cdot$}
\put(0.5,20.02){$\cdot$}
\put(30,0.5){$\cdot$} \put(30.5,0.61){$\cdot$}
\put(31,0.73){$\cdot$} \put(31.5,0.86){$\cdot$} \put(32,1){$\cdot$}
\put(32.5,1.15){$\cdot$} \put(33,1.31){$\cdot$}
\put(33.5,1.49){$\cdot$} \put(34,1.68){$\cdot$}
\put(34.5,1.87){$\cdot$} \put(35,2.09){$\cdot$}
\put(35.5,2.31){$\cdot$} \put(36,2.55){$\cdot$}
\put(36.5,2.8){$\cdot$} \put(37,3.06){$\cdot$}
\put(37.5,3.35){$\cdot$} \put(38,3.64){$\cdot$}
\put(38.5,3.96){$\cdot$} \put(39,4.29){$\cdot$}
\put(39.5,4.63){$\cdot$} \put(40,5){$\cdot$}
\put(40.5,5.38){$\cdot$} \put(41,5.79){$\cdot$}
\put(41.5,6.22){$\cdot$} \put(42,6.67){$\cdot$}
\put(42.5,7.15){$\cdot$} \put(43,7.65){$\cdot$}
\put(43.5,8.18){$\cdot$} \put(44,8.75){$\cdot$}
\put(44.5,9.35){$\cdot$} \put(45,10){$\cdot$}
\put(45.5,10.69){$\cdot$} \put(46,11.43){$\cdot$}
\put(46.5,12.24){$\cdot$} \put(47,13.12){$\cdot$}
\put(47.5,14.10){$\cdot$} \put(48,15.20){$\cdot$}
\put(48.5,16.47){$\cdot$} \put(49,18){$\cdot$}
\put(49.5,20.02){$\cdot$}

\end{picture}
\]
\vspace{2ex}
\begin{center}
\begin{minipage}{25pc}{\small
{\bf Fig.\,5:} An action of the supergroup
$BW_{\C}=G(\C_n,\gamma,\bcirc)$ on the system $\fM=\fM^0\oplus\fM^1$
of complex representations $\fC$ of $\fG_+\simeq\spin_+(1,3)$.}
\end{minipage}
\end{center}
\end{figure}
Further, in virtue of the isomorphism $\C^+_4\simeq\C_3$ the
transition $\C^+_4\rightarrow\C_4$ ($\C_3\rightarrow\C_4$) induces
on the system $\fM$ the following relation:
$\overset{+}{\fC}\!{}^{2,0}\rightarrow\fC^{2,0}$ or
$\fC^{1,0}\rightarrow\fC^{2,0}$ and, therefore, $h=1$. Thus, a full
rotation (cycle), presented on the diagram Fig.\,5, corresponds to a
transition from the representation $\fC^{1,0}$ of the spin $1/2$,
associated with the algebra $\C_2$, to the representation
$\fC^{2,0}$ of the spin 1, associated with the algebra $\C_4$. It is
easy to see that a sequence of transitions $\C^+_5\rightarrow\C_5$
($\C_4\rightarrow\C_5$) and $\C^+_6\rightarrow\C_6$
($\C_5\rightarrow\C_6$) leads to a transition from the
representation $\fC^{2,0}$ of the spin 1 to the representation
$\fC^{3,0}$ of the spin $3/2$, associated with the algebra $\C_6$,
and so on. Thus, the action of the supergroup
$BW_{\C}=G(\C_n,\gamma,\bcirc)$ on the system $\fM=\fM^0\oplus\fM^1$
converts fermionic representations into bosonic representations and
vice versa, that is, this action is equivalent to a supersymmetric
action.

It is obvious that a cyclic structure over the algebras $\C_n$,
defined by the supergroup $BW_{\C}$, is related immediately with a
modulo 2 periodicity of complex Clifford algebras
\cite{AtBSh,Kar79}: $\C_{n+2}\simeq\C_n\otimes\C_2$. Therefore, we
have the following relations for the complex representations $\fC$
and $\overset{\ast}{\fC}$:
\[
\fC^{l_0+l_1,0}\;\simeq\;\fC^{l_0+l_1-1,0}\otimes\fC^{1,0},
\]
\[
\fC^{0,l_0-l_1}\;\simeq\;\fC^{0,l_0-l_1+1}\otimes\fC^{0,-1}.
\]
Correspondingly, for the representation
$\fC\otimes\overset{\ast}{\fC}$ we obtain
\[
\fC^{l_0+l_1,l_0-l_1}\simeq\fC^{l_0+l_1-1,l_0-l_1+1}\otimes\fC^{1,-1}.
\]
Thus, the action of the supergroup $BW_{\C}$ forms a cycle of the
period 2 on the system $\fM$, where a basic cycle generating factor
is the fundamental representation $\fC^{1,0}$ ($\fC^{0,-1}$) of the
group $\fG_+$.\\
2) Real representations.\\
A restriction of the complex representations $\fC$ onto the real
representations $\fR$ and $\fH$ leads to a much more high-graded
periodicity over the field $\F=\R$. Physical fields, defined within
such representations, describe neutral particles, or particles
occurring in the rest such as atomic nuclei. So, the Clifford
algebra $\cl_{p,q}$ is central simple if $p-q\not\equiv
1,5\pmod{8}$. Graded central simple Clifford algebras over the field
$\F=\R$ form eight similarity classes, which, as it is easy to see,
coincide with the eight types of the algebras $\cl_{p,q}$. The set
of these eight types (classes) forms the Brauer-Wall supergroup
$BW_{\R}=G(\cl_{p,q},\gamma,\bcirc)$ \cite{Wal64,Lou81}, a cyclic
structure of which is formally equivalent to the action of cyclic
group $\dZ_8$ (see section 2). Therefore, in virtue of the
correspondences (\ref{IdentPer}) and (\ref{IdentPerD}) an action of
$BW_{\R}$ can be transferred onto the system $\fM=\fM^+\oplus\fM^-$.
In its turn, a cyclic structure of the supergroup $BW_{\R}$ is
defined by a transition
$\cl^+_{p,q}\overset{h}{\longrightarrow}\cl_{p,q}$, where the type
of $\cl_{p,q}$ is defined by the formula $q-p=h+8r$, here
$h\in\{0,\ldots,7\}$, $r\in\dZ$ \cite{BTr87,BT88}. Thus, the action
of $BW_{\R}$ on the system $\fM$ is defined by a transition
$\overset{+}{\fD}\!{}^{l_0}\overset{h}{\longrightarrow}\fD^{l_0}$,
where
$\fD^{\l_0}=\left\{\fR^{l_0}_{0,2},\fH^{l_0}_{4,6},\fC^{l_0}_{3,7},
\fR^{l_0}_{0,2}\cup\fR^{l_0}_{0,2},\fH^{l_0}_{4,6}\cup\fH^{l_0}_{4,6}\right\}$,
and $\overset{+}{\fD}\!{}^{r/2}\simeq\fD^{\frac{r-1}{2}}$ if
$\fD\in\fM^+$ and
$\left(\fD^{r/2}\cup\fD^{r/2}\right)^+\simeq\fD^{r/2}$ if
$\fD\in\fM^-$. Therefore, the cyclic structure of the supergroup
$BW_{\R}=G(\cl_{p,q},\gamma,\bcirc)$ induces on the system $\fM$
modulo 8 periodic relations, which can be shown explicitly on the
diagram Fig.\,6.
\begin{figure}
\[
\unitlength=0.5mm
\begin{picture}(100.00,110.00)

\put(97,67){$\fC^{l_0}_7$}\put(108,64){$p-q\equiv
7\!\!\!\!\pmod{8}$} \put(80,80){1} \put(75,93.3){$\cdot$}
\put(75.5,93){$\cdot$} \put(76,92.7){$\cdot$}
\put(76.5,92.4){$\cdot$} \put(77,92.08){$\cdot$}
\put(77.5,91.76){$\cdot$} \put(78,91.42){$\cdot$}
\put(78.5,91.08){$\cdot$} \put(79,90.73){$\cdot$}
\put(79.5,90.37){$\cdot$} \put(80,90.0){$\cdot$}
\put(80.5,89.62){$\cdot$} \put(81,89.23){$\cdot$}
\put(81.5,88.83){$\cdot$} \put(82,88.42){$\cdot$}
\put(82.5,87.99){$\cdot$} \put(83,87.56){$\cdot$}
\put(83.5,87.12){$\cdot$} \put(84,86.66){$\cdot$}
\put(84.5,86.19){$\cdot$} \put(85,85.70){$\cdot$}
\put(85.5,85.21){$\cdot$} \put(86,84.69){$\cdot$}
\put(86.5,84.17){$\cdot$} \put(87,83.63){$\cdot$}
\put(87.5,83.07){$\cdot$} \put(88,82.49){$\cdot$}
\put(88.5,81.9){$\cdot$} \put(89,81.29){$\cdot$}
\put(89.5,80.65){$\cdot$} \put(90,80){$\cdot$}
\put(90.5,79.32){$\cdot$} \put(91,78.62){$\cdot$}
\put(91.5,77.89){$\cdot$} \put(92,77.13){$\cdot$}
\put(92.5,76.34){$\cdot$} \put(93,75.51){$\cdot$}
\put(93.5,74.65){$\cdot$} \put(94,73.74){$\cdot$}
\put(94.5,72.79){$\cdot$} \put(96.5,73.74){\vector(1,-2){1}}
\put(80,20){3} \put(97,31){$\fH^{l_0}_6$}\put(108,28){$p-q\equiv
6\!\!\!\!\pmod{8}$} \put(75,6.7){$\cdot$} \put(75.5,7){$\cdot$}
\put(76,7.29){$\cdot$} \put(76.5,7.6){$\cdot$}
\put(77,7.91){$\cdot$} \put(77.5,8.24){$\cdot$}
\put(78,8.57){$\cdot$} \put(78.5,8.91){$\cdot$}
\put(79,9.27){$\cdot$} \put(79.5,9.63){$\cdot$} \put(80,10){$\cdot$}
\put(80.5,10.38){$\cdot$} \put(81,10.77){$\cdot$}
\put(81.5,11.17){$\cdot$} \put(82,11.58){$\cdot$}
\put(82.5,12.00){$\cdot$} \put(83,12.44){$\cdot$}
\put(83.5,12.88){$\cdot$} \put(84,13.34){$\cdot$}
\put(84.5,13.8){$\cdot$} \put(85,14.29){$\cdot$}
\put(85.5,14.79){$\cdot$} \put(86,15.3){$\cdot$}
\put(86.5,15.82){$\cdot$} \put(87,16.37){$\cdot$}
\put(87.5,16.92){$\cdot$} \put(88,17.5){$\cdot$}
\put(88.5,18.09){$\cdot$} \put(89,18.71){$\cdot$}
\put(89.5,19.34){$\cdot$} \put(90,20){$\cdot$}
\put(90.5,20.68){$\cdot$} \put(91,21.38){$\cdot$}
\put(91.5,22.11){$\cdot$} \put(92,22.87){$\cdot$}
\put(92.5,23.66){$\cdot$} \put(93,24.48){$\cdot$}
\put(93.5,25.34){$\cdot$} \put(94,26.25){$\cdot$}
\put(94.5,27.20){$\cdot$} \put(95,28.20){$\cdot$}
\put(20,80){7} \put(25,93.3){$\cdot$} \put(24.5,93){$\cdot$}
\put(24,92.7){$\cdot$} \put(23.5,92.49){$\cdot$}
\put(23,92.08){$\cdot$} \put(22.5,91.75){$\cdot$}
\put(22,91.42){$\cdot$} \put(21.5,91.08){$\cdot$}
\put(21,90.73){$\cdot$} \put(20.5,90.37){$\cdot$}
\put(20,90){$\cdot$} \put(19.5,89.62){$\cdot$}
\put(19,89.23){$\cdot$} \put(18.5,88.83){$\cdot$}
\put(18,88.42){$\cdot$} \put(17.5,87.99){$\cdot$}
\put(17,87.56){$\cdot$} \put(16.5,87.12){$\cdot$}
\put(16,86.66){$\cdot$} \put(15.5,86.19){$\cdot$}
\put(15,85.70){$\cdot$} \put(14.5,85.21){$\cdot$}
\put(14,84.69){$\cdot$} \put(13.5,84.17){$\cdot$}
\put(13,83.63){$\cdot$} \put(12.5,83.07){$\cdot$}
\put(12,82.49){$\cdot$} \put(11.5,81.9){$\cdot$}
\put(11,81.29){$\cdot$} \put(10.5,80.65){$\cdot$}
\put(10,80){$\cdot$} \put(9.5,79.32){$\cdot$} \put(9,78.62){$\cdot$}
\put(8.5,77.89){$\cdot$} \put(8,77.13){$\cdot$}
\put(7.5,76.34){$\cdot$} \put(7,75.51){$\cdot$}
\put(6.5,74.65){$\cdot$} \put(6,73.79){$\cdot$}
\put(5.5,72.79){$\cdot$} \put(5,71.79){$\cdot$}
\put(20,20){5} \put(25,6.7){$\cdot$} \put(24.5,7){$\cdot$}
\put(24,7.29){$\cdot$} \put(23.5,7.6){$\cdot$}
\put(23,7.91){$\cdot$} \put(22.5,8.24){$\cdot$}
\put(22,8.57){$\cdot$} \put(21.5,8.91){$\cdot$}
\put(21,9.27){$\cdot$} \put(20.5,9.63){$\cdot$} \put(20,10){$\cdot$}
\put(19.5,10.38){$\cdot$} \put(19,10.77){$\cdot$}
\put(18.5,11.17){$\cdot$} \put(18,11.58){$\cdot$}
\put(17.5,12){$\cdot$} \put(17,12.44){$\cdot$}
\put(16.5,12.88){$\cdot$} \put(16,13.34){$\cdot$}
\put(15.5,13.8){$\cdot$} \put(15,14.29){$\cdot$}
\put(14.5,14.79){$\cdot$} \put(14,15.3){$\cdot$}
\put(13.5,15.82){$\cdot$} \put(13,16.37){$\cdot$}
\put(12.5,16.92){$\cdot$} \put(12,17.5){$\cdot$}
\put(11.5,18.09){$\cdot$} \put(11,18.71){$\cdot$}
\put(10.5,19.34){$\cdot$} \put(10,20){$\cdot$}
\put(9.5,20.68){$\cdot$} \put(9,21.38){$\cdot$}
\put(8.5,22.11){$\cdot$} \put(8,22.87){$\cdot$}
\put(7.5,23.66){$\cdot$} \put(7,24.48){$\cdot$}
\put(6.5,25.34){$\cdot$} \put(6,26.25){$\cdot$}
\put(5.5,27.20){$\cdot$} \put(5,28.20){$\cdot$}
\put(-1,97){$\fR^{l_0}_{0,2}\cup\fR^{l_0}_{0,2}$}
\put(-55,107){$p-q\equiv 1\!\!\!\!\pmod{8}$} \put(50,93){0}
\put(50,100){$\cdot$} \put(49.5,99.99){$\cdot$}
\put(49,99.98){$\cdot$} \put(48.5,99.97){$\cdot$}
\put(48,99.96){$\cdot$} \put(47.5,99.94){$\cdot$}
\put(47,99.91){$\cdot$} \put(46.5,99.86){$\cdot$}
\put(46,99.84){$\cdot$} \put(45.5,99.8){$\cdot$}
\put(45,99.75){$\cdot$} \put(44.5,99.7){$\cdot$}
\put(44,99.64){$\cdot$} \put(43.5,99.57){$\cdot$}
\put(43,99.51){$\cdot$} \put(42.5,99.43){$\cdot$}
\put(42,99.35){$\cdot$} \put(41.5,99.27){$\cdot$}
\put(41,99.18){$\cdot$} \put(40.5,99.09){$\cdot$}
\put(40,98.99){$\cdot$} \put(39.5,98.88){$\cdot$}
\put(39,98.77){$\cdot$} \put(38.5,98.66){$\cdot$}
\put(38,98.54){$\cdot$} \put(37.5,98.41){$\cdot$}
\put(37,98.28){$\cdot$} \put(50.5,99.99){$\cdot$}
\put(51,99.98){$\cdot$} \put(51.5,99.97){$\cdot$}
\put(52,99.96){$\cdot$} \put(52.5,99.94){$\cdot$}
\put(53,99.91){$\cdot$} \put(53.5,99.86){$\cdot$}
\put(54,99.84){$\cdot$} \put(54.5,99.8){$\cdot$}
\put(55,99.75){$\cdot$} \put(55.5,99.7){$\cdot$}
\put(56,99.64){$\cdot$} \put(56.5,99.57){$\cdot$}
\put(57,99.51){$\cdot$} \put(57.5,99.43){$\cdot$}
\put(58,99.35){$\cdot$} \put(58.5,99.27){$\cdot$}
\put(59,99.18){$\cdot$} \put(59.5,99.09){$\cdot$}
\put(60,98.99){$\cdot$} \put(60.5,98.88){$\cdot$}
\put(61,98.77){$\cdot$} \put(61.5,98.66){$\cdot$}
\put(62,98.54){$\cdot$} \put(62.5,98.41){$\cdot$}
\put(63,98.28){$\cdot$}
\put(65,97){$\fR^{l_0}_0$}\put(73,108){$p-q\equiv
0\!\!\!\!\pmod{8}$}
\put(50,7){4} \put(67,2){$\fH^{l_0}_{4,6}\cup\fH^{l_0}_{4,6}$}
\put(90,-6){$p-q\equiv 5\!\!\!\!\pmod{8}$} \put(50,0){$\cdot$}
\put(50.5,0){$\cdot$} \put(51,0.01){$\cdot$}
\put(51.5,0.02){$\cdot$} \put(52,0.04){$\cdot$}
\put(52.5,0.06){$\cdot$} \put(53,0.09){$\cdot$}
\put(53.5,0.12){$\cdot$} \put(54,0.16){$\cdot$}
\put(54.5,0.2){$\cdot$} \put(55,0.25){$\cdot$}
\put(55.5,0.3){$\cdot$} \put(56,0.36){$\cdot$}
\put(56.5,0.42){$\cdot$} \put(57,0.49){$\cdot$}
\put(57.5,0.56){$\cdot$} \put(58,0.64){$\cdot$}
\put(58.5,0.73){$\cdot$} \put(59,0.82){$\cdot$}
\put(59.5,0.91){$\cdot$} \put(60,1.01){$\cdot$}
\put(60.5,1.11){$\cdot$} \put(61,1.22){$\cdot$}
\put(61.5,1.34){$\cdot$} \put(62,1.46){$\cdot$}
\put(62.5,1.59){$\cdot$} \put(63,1.72){$\cdot$}
\put(49.5,0){$\cdot$} \put(49,0.01){$\cdot$}
\put(48.5,0.02){$\cdot$} \put(48,0.04){$\cdot$}
\put(47.5,0.06){$\cdot$} \put(47,0.09){$\cdot$}
\put(46.5,0.12){$\cdot$} \put(46,0.16){$\cdot$}
\put(45.5,0.2){$\cdot$} \put(45,0.25){$\cdot$}
\put(44.5,0.3){$\cdot$} \put(44,0.36){$\cdot$}
\put(43.5,0.42){$\cdot$} \put(43,0.49){$\cdot$}
\put(42.5,0.56){$\cdot$} \put(42,0.64){$\cdot$}
\put(41.5,0.73){$\cdot$} \put(41,0.82){$\cdot$}
\put(40.5,0.91){$\cdot$} \put(40,1.01){$\cdot$}
\put(39.5,1.11){$\cdot$} \put(39,1.22){$\cdot$}
\put(38.5,1.34){$\cdot$} \put(38,1.46){$\cdot$}
\put(37.5,1.59){$\cdot$} \put(37,1.72){$\cdot$}
\put(28,3){$\fH^{l_0}_4$}\put(-40,-4){$p-q\equiv 4\!\!\!\!\pmod{8}$}
\put(93,50){2} \put(98.28,63){$\cdot$} \put(98.41,62.5){$\cdot$}
\put(98.54,62){$\cdot$} \put(98.66,61.5){$\cdot$}
\put(98.77,61){$\cdot$} \put(98.88,60.5){$\cdot$}
\put(98.99,60){$\cdot$} \put(99.09,59.5){$\cdot$}
\put(99.18,59){$\cdot$} \put(99.27,58.5){$\cdot$}
\put(99.35,58){$\cdot$} \put(99.43,57.5){$\cdot$}
\put(99.51,57){$\cdot$} \put(99.57,56.5){$\cdot$}
\put(99.64,56){$\cdot$} \put(99.7,55.5){$\cdot$}
\put(99.75,55){$\cdot$} \put(99.8,54.5){$\cdot$}
\put(99.84,54){$\cdot$} \put(99.86,53.5){$\cdot$}
\put(99.91,53){$\cdot$} \put(99.94,52.5){$\cdot$}
\put(99.96,52){$\cdot$} \put(99.97,51.5){$\cdot$}
\put(99.98,51){$\cdot$} \put(99.99,50.5){$\cdot$}
\put(100,50){$\cdot$} \put(98.28,37){$\cdot$}
\put(98.41,37.5){$\cdot$} \put(98.54,38){$\cdot$}
\put(98.66,38.5){$\cdot$} \put(98.77,39){$\cdot$}
\put(98.88,39.5){$\cdot$} \put(98.99,40){$\cdot$}
\put(99.09,40.5){$\cdot$} \put(99.18,41){$\cdot$}
\put(99.27,41.5){$\cdot$} \put(99.35,42){$\cdot$}
\put(99.43,42.5){$\cdot$} \put(99.51,43){$\cdot$}
\put(99.57,43.5){$\cdot$} \put(99.64,44){$\cdot$}
\put(99.7,44.5){$\cdot$} \put(99.75,45){$\cdot$}
\put(99.8,45.5){$\cdot$} \put(99.84,46){$\cdot$}
\put(99.86,46.5){$\cdot$} \put(99.91,47){$\cdot$}
\put(99.94,47.5){$\cdot$} \put(99.96,48){$\cdot$}
\put(99.97,48.5){$\cdot$} \put(99.98,49){$\cdot$}
\put(99.99,49.5){$\cdot$}
\put(7,50){6} \put(1,32){$\fC^{l_0}_3$}\put(-65,29){$p-q\equiv
3\!\!\!\!\pmod{8}$} \put(1.72,63){$\cdot$} \put(1.59,62.5){$\cdot$}
\put(1.46,62){$\cdot$} \put(1.34,61.5){$\cdot$}
\put(1.22,61){$\cdot$} \put(1.11,60.5){$\cdot$}
\put(1.01,60){$\cdot$} \put(0.99,59.5){$\cdot$}
\put(0.82,59){$\cdot$} \put(0.73,58.5){$\cdot$}
\put(0.64,58){$\cdot$} \put(0.56,57.5){$\cdot$}
\put(0.49,57){$\cdot$} \put(0.42,56.5){$\cdot$}
\put(0.36,56){$\cdot$} \put(0.3,55.5){$\cdot$}
\put(0.25,55){$\cdot$} \put(0.2,54.5){$\cdot$}
\put(0.16,54){$\cdot$} \put(0.12,53.5){$\cdot$}
\put(0.09,53){$\cdot$} \put(0.06,52.5){$\cdot$}
\put(0.04,52){$\cdot$} \put(0.02,51.5){$\cdot$}
\put(0.01,51){$\cdot$} \put(0,50.5){$\cdot$} \put(0,50){$\cdot$}
\put(1.72,37){$\cdot$} \put(1.59,37.5){$\cdot$}
\put(1.46,38){$\cdot$} \put(1.34,38.5){$\cdot$}
\put(1.22,39){$\cdot$} \put(1.11,39.5){$\cdot$}
\put(1.01,40){$\cdot$} \put(0.99,40.5){$\cdot$}
\put(0.82,41){$\cdot$} \put(0.73,41.5){$\cdot$}
\put(0.64,42){$\cdot$} \put(0.56,42.5){$\cdot$}
\put(0.49,43){$\cdot$} \put(0.42,43.5){$\cdot$}
\put(0.36,44){$\cdot$} \put(0.3,44.5){$\cdot$}
\put(0.25,45){$\cdot$} \put(0.2,45.5){$\cdot$}
\put(0.16,46){$\cdot$} \put(0.12,46.5){$\cdot$}
\put(0.09,47){$\cdot$} \put(0.06,47.5){$\cdot$}
\put(0.04,48){$\cdot$} \put(0.02,48.5){$\cdot$}
\put(0.01,49){$\cdot$} \put(0,49.5){$\cdot$}
\put(0.5,67){$\fR^{l_0}_2$}\put(-65,75){$p-q\equiv
2\!\!\!\!\pmod{8}$}
\end{picture}
\]

\vspace{2ex}
\begin{center}
\begin{minipage}{25pc}{\small
{\bf Fig.\,6:} An action of the supergroup
$BW_{\R}=G(\cl_{p,q},\gamma,\bcirc)$ on the system
$\fM=\fM^+\oplus\fM^-$ of real representations $\fD$ of
$\fG_+\simeq\spin_+(1,3)$, $l_0=\frac{p+q}{4}$.}
\end{minipage}
\end{center}
\end{figure}
Let us consider in detail one action cycle of the supergroup
$BW_{\R}$ on the system $\fM=\fM^+\oplus\fM^-$. In virtue of an
isomorphism $\cl^+_{1,1}\simeq\cl_{1,0}$ a transition
$\cl^+_{1,1}\rightarrow\cl_{1,1}$ ($\cl_{1,0}\rightarrow\cl_{1,1}$)
induces on the system $\fM$ a transition
$\fR^0_0\cup\fR^0_0\rightarrow\fR^{\frac{1}{2}}_0$, where
$\cl_{1,0}\simeq\cl_{0,0}\cup\cl_{0,0}$ is an algebra of double
numbers, $\fR^0_0$ is the field of real numbers,
$\fR^{\frac{1}{2}}_0$ is a real representation of the spin $1/2$,
associated with the algebra $\cl_{1,1}$. At this point, $h=0$ and
$r=0$ (an original point of the first cycle). Further, in virtue of
$\cl^+_{1,2}\simeq\cl_{1,1}$ a transition
$\cl^+_{1,2}\rightarrow\cl_{1,2}$ ($\cl_{1,1}\rightarrow\cl_{1,2}$)
induces on the system $\fM$ a transition
$\fR^{\frac{1}{2}}_0\rightarrow\fR^{\frac{1}{2}}_0\cup\fR^{\frac{1}{2}}_0$,
at this point, $h=1$ and $r=0$. The following transition
$\cl^+_{1,3}\rightarrow\cl_{1,3}$ ($\cl_{1,2}\rightarrow\cl_{1,3}$)
in virtue of $\cl^+_{1,3}\simeq\cl_{1,2}$ leads to a transition
$\fR^{\frac{1}{2}}_0\cup\fR^{\frac{1}{2}}_0\rightarrow\fH^1_6$,
where $\fH^1_6$ is a quaternionic representation of the spin 1,
associated with the algebra $\cl_{1,3}$. At this transition we have
$h=2$ and $r=0$. In virtue of an isomorphism
$\cl^+_{1,4}\simeq\cl_{1,3}$ a transition
$\cl^+_{1,4}\rightarrow\cl_{1,4}$ ($\cl_{1,3}\rightarrow\cl_{1,4}$)
induces on the system $\fM$ a transition
$\fH^1_6\rightarrow\fH^1_6\cup\fH^1_6$, at this point, $h=3$ and
$r=0$. Further, in virtue of $\cl^+_{1,5}\simeq\cl_{1,4}$ a
transition $\cl^+_{1,5}\rightarrow\cl_{1,5}$
($\cl_{1,4}\rightarrow\cl_{1,5}$) induces on the system $\fM$ a
transition $\fH^1_6\cup\fH^1_6\rightarrow\fH^{\frac{3}{2}}_4$, where
$\fH^{\frac{3}{2}}_4$ is a quaternionic representation of the spin
$3/2$, associated with the algebra $\cl_{1,5}$. At this transition
we have $h=4$ and $r=0$. The following transition
$\cl^+_{1,6}\rightarrow\cl_{1,6}$ ($\cl_{1,5}\rightarrow\cl_{1,6}$)
in virtue of $\cl^+_{1,6}\simeq\cl_{1,5}$ leads to a transition
$\fH^{\frac{3}{2}}_4\rightarrow\fH^{\frac{3}{2}}_4\cup\fH^{\frac{3}{2}}_4$,
at this point, $h=5$ and $r=0$. The transition
$\cl^+_{1,7}\rightarrow\cl_{1,7}$ ($\cl_{1,6}\rightarrow\cl_{1,7}$)
in virtue of $\cl^+_{1,7}\simeq\cl_{1,6}$ induces on the system
$\fM$ a transition
$\fH^{\frac{3}{2}}_4\cup\fH^{\frac{3}{2}}_4\rightarrow\fR^2_2$,
where $\fR^2_2$ is a real representation of the spin 2, associated
with the algebra $\cl_{1,7}$. At this transition we have $h=6$ and
$r=0$. Finally, a transition $\cl^+_{1,8}\rightarrow\cl_{1,8}$
($\cl_{1,7}\rightarrow\cl_{1,8}$) finishes the first cycle ($h=7$,
$r=0$) and in virtue of $\cl^+_{1,8}\simeq\cl_{1,7}$ induces on the
system $\fM$ the following transition:
$\fR^2_2\rightarrow\fR^2_2\cup\fR^2_2$. It is obvious that a new
cycle ($h=0$, $r=1$) is started with a transition
$\cl^+_{1,9}\rightarrow\cl_{1,9}$ ($\cl_{1,8}\rightarrow\cl_{1,9}$)
and so on.

Further, a cyclic structure over the algebras $\cl_{p,q}$, defined
by the supergroup $BW_{\R}$, is related immediately with the
Atiyah-Bott-Shapiro periodicity \cite{AtBSh}. In accordance with
\cite{AtBSh} the Clifford algebra over the field $\F=\R$ is modulo 8
periodic:
$\cl_{p+8,q}\simeq\cl_{p,q}\otimes\cl_{8,0}\,(\cl_{p,q+8}\simeq\cl_{p,q}
\otimes\cl_{0,8})$. Therefore, the following relation takes place:
\[
\fD^{l_0+2}\simeq\fD^{l_0}\otimes\fR^2_0,
\]
since $\fR^2_0\leftrightarrow\cl_{8,0}\;(\cl_{0,8})$. In virtue of
the theorem 2 from (\ref{Ten}) it follows that
$\cl_{8,0}\simeq\cl_{2,0}\otimes\cl_{0,2}\otimes\cl_{0,2}\otimes\cl_{2,0}$
($\cl_{0,8}\simeq\cl_{0,2}\otimes\cl_{2,0}\otimes\cl_{2,0}
\otimes\cl_{0,2}$). A minimal left ideal of the algebra $\cl_{8,0}$
is the spinspace $\dS_{16}$, which in virtue of the real ring
$\K\simeq\R$ is defined within a full matrix algebra
$\Mat_{16}(\R)$. At first glance, from the factorization of the
algebra $\cl_{8,0}$ it follows that
$\Mat_2(\R)\otimes\BH\otimes\BH\otimes\Mat_2(\R)\not\simeq\Mat_{16}(\R)$,
but it is wrong, since there exists an isomorphism
$\BH\otimes\BH\simeq\Mat_4(\R)$ (see Appendix B in \cite{BDGK}). On
the other hand, in terms of minimal left ideals the modulo 8
periodicity is defined as follows:
\[
\dS_{n+8}\simeq\dS_n\otimes\dS_{16}.
\]
In virtue of the mapping
$\gamma_{8,0}:\,\cl_{8,0}\rightarrow\Mat_2(\dO)$ \cite{MS96,Bae01},
the latter relation can be written in the form
\[
\dS_{n+8}\simeq\dS_n\otimes\dO^2,
\]
where $\dO$ is {\it an octonion algebra}. Because the algebra
$\cl_{8,0}\simeq \cl_{0,8}$ admits an octonionic representation,
then in virtue of the modulo 8 periodicity octonionic
representations can be defined for all high dimensions and,
therefore, on the system $\fM=\fM^+\oplus\fM^-$ of real
representations of the group $\fG_+$ we have a relation
\[
\fD^{l_0+2}\simeq\fD^{l_0}\otimes\fO,
\]
where $\fO$ is {\it an octonionic representation} of the group
$\fG_+=\spin_+(1,3)$ ($\fO\sim\fR^0_4$).

\section*{Acknowledgements}
I am deeply grateful to Prof. C. T. C. Wall  and Prof. A. Trautman
for sending me their interesting works.

\end{document}